\documentclass[12pt]{article}
\usepackage{amsmath}
\usepackage{algorithm}
\usepackage{amsfonts}
\usepackage{amsthm}
\usepackage{float}
\usepackage{graphicx}
\usepackage{enumerate}
\usepackage{amssymb}
\usepackage{mathtools}
\usepackage{natbib}
\usepackage{url} % not crucial - just used below for the URL 

\newcommand{\blind}{1}

\usepackage{bm}
\DeclareBoldMathCommand{\ba}{a}
\DeclareBoldMathCommand{\bb}{b}
\DeclareBoldMathCommand{\be}{e}
\DeclareBoldMathCommand{\br}{r}
\DeclareBoldMathCommand{\bg}{g}
\DeclareBoldMathCommand{\bt}{t}
\DeclareBoldMathCommand{\bu}{u}
\DeclareBoldMathCommand{\bv}{v}
\DeclareBoldMathCommand{\be}{e}
\DeclareBoldMathCommand{\bw}{w}
\DeclareBoldMathCommand{\bx}{x}
\DeclareBoldMathCommand{\bz}{z}
\DeclareBoldMathCommand{\bm}{m}
\DeclareBoldMathCommand{\by}{y}
\DeclareBoldMathCommand{\bh}{h}

\DeclareBoldMathCommand{\bA}{A}
\DeclareBoldMathCommand{\bD}{D}

\DeclareBoldMathCommand{\bY}{Y}
\DeclareBoldMathCommand{\bX}{X}
\DeclareBoldMathCommand{\bZ}{Z}
\DeclareBoldMathCommand{\bM}{M}
\DeclareBoldMathCommand{\bP}{P}
\DeclareBoldMathCommand{\bI}{I}
\DeclareBoldMathCommand{\bT}{T}
\DeclareBoldMathCommand{\bU}{U}
\DeclareBoldMathCommand{\bS}{S}
\DeclareBoldMathCommand{\bQ}{Q}
\DeclareBoldMathCommand{\bW}{W}
\DeclareBoldMathCommand{\bE}{E}
\DeclareBoldMathCommand{\bV}{V}
\DeclareBoldMathCommand{\bJ}{J}
\DeclareBoldMathCommand{\bL}{L}
\DeclareBoldMathCommand{\bR}{R}

\DeclareBoldMathCommand{\bzero}{0}
\DeclareBoldMathCommand{\bone}{1}
\DeclareBoldMathCommand{\balpha}{\alpha}
\DeclareBoldMathCommand{\bxi}{\xi}
\DeclareBoldMathCommand{\bbeta}{\beta}
\DeclareBoldMathCommand{\brho}{\rho}
\DeclareBoldMathCommand{\btau}{\tau}
\DeclareBoldMathCommand{\beeta}{\eta}
\DeclareBoldMathCommand{\btheta}{\theta}
\DeclareBoldMathCommand{\bdelta}{\delta}
\DeclareBoldMathCommand{\bPhi}{\Phi}
\DeclareBoldMathCommand{\bkappa}{\kappa}
\DeclareBoldMathCommand{\bgamma}{\gamma}
\DeclareBoldMathCommand{\bSigma}{\Sigma}
\DeclareBoldMathCommand{\bmathcaly}{\mathcal{Y}}
\DeclareBoldMathCommand{\bTheta}{\Theta}
\DeclareBoldMathCommand{\bmu}{\mu}
\DeclareBoldMathCommand{\bpsi}{\psi}
\DeclareBoldMathCommand{\bzeta}{\zeta}
\DeclareBoldMathCommand{\bvarepsilon}{\varepsilon}
\DeclareBoldMathCommand{\bepsilon}{\epsilon}

\DeclareBoldMathCommand{\rn}{\frac{\phi}{\phi+\|\tilde{\bZ}_n\|_2^2}}

\def\bTheta{\boldsymbol{\Theta}}

\def\bSigma{\boldsymbol{\Sigma}}

\DeclareMathOperator*{\argmin}{argmin}

\newcommand{\Rp}{\mathbb{R}_{\geq 0}}

\newtheorem{theorem}{Theorem}[section]

\newtheorem{lemma}[theorem]{Lemma}
\newtheorem{corollary}[theorem]{Corollary}

\newtheorem{example}[theorem]{Example}

\begin{document}

\def\spacingset#1{\renewcommand{\baselinestretch}%
{#1}\small\normalsize} \spacingset{1}

%%%%%%%%%%%%%%%%%%%%%%%%%%%%%%%%%%%%%%%%%%%%%%%%%%%%%%%%%%%%%%%%%%%%%%%%%%%%%%

\if1\blind
{
  \title{\bf Anytime-Valid Continuous-Time Confidence Processes for Inhomogeneous Poisson Processes}
  \author{Michael Lindon\footnote{Netflix}, Nathan Kallus\footnote{Netflix, Cornell}}
  \maketitle
} \fi

\if0\blind
{
  \bigskip
  \bigskip
  \bigskip
  \begin{center}
    {\LARGE\bf Sub-Poisson Point Process Intensity Estimation}
\end{center}
  \medskip
} \fi

\bigskip
\begin{abstract}
Motivated by monitoring the arrival of incoming adverse events such as customer support calls or crash reports from users exposed to an experimental product change, we consider sequential hypothesis testing of continuous-time inhomogeneous Poisson point processes. Specifically, we provide an interval-valued confidence process $C^\alpha(t)$ over continuous time $t$ for the cumulative arrival rate $\Lambda(t) = \int_0^t \lambda(s) \mathrm{d}s$ with a continuous-time anytime-valid coverage guarantee $\mathbb{P}[\Lambda(t) \in C^\alpha(t) \, \forall t >0] \geq 1-\alpha$. We extend our results to compare two independent arrival processes by constructing multivariate confidence processes and a closed-form $e$-process for testing the equality of rates with a time-uniform Type-I error guarantee at a nominal $\alpha$. We characterize the asymptotic growth rate of the proposed $e$-process under the alternative and show that it has power 1 when the average rates of the two Poisson process differ in the limit. We also observe a complementary relationship between our multivariate confidence process and the universal inference $e$-process for testing composite null hypotheses.
\end{abstract}

\noindent%
\vfill
\footnotesize
\qquad 
\\

\normalsize
\newpage
\spacingset{1.9} % DON'T change the spacing!
\section{Introduction}
Point processes are frequently encountered across many industry applications dealing with random arrival times of events. Consider, for example, the arrival of telephone calls at customer support centers \citep{brown2005statistical} or the epochs of failures in software reliability testing \citep{yang}. Of particular interest is how a proposed change, such as a software update, might impact the volume of arrivals, which is often explored through an A/B test before deciding whether to deploy the change. We are here motivated by the need to detect as soon as possible if the proposed change does increase the volume of arrivals of adverse events, rather than, for example, analyzing the difference in event volumes at the end of a pre-specified window, which may be either too short or too long.

Complicating this question is that instantaneous arrival rates, with or without the change, can generally fluctuate significantly over time due to usage patterns, external factors, and following the onset of service outages. The inhomogeneous Poisson point process (defined in section \ref{sec:review}) is a popular model for describing such phenomena \citep{kingman}. A time-varying intensity function $\lambda(t)$ allows the instantaneous arrival rate to vary flexibly throughout time, while the memoryless and independent increments properties make inference and computation tractable.

In this paper, we study anytime-valid inference on the cumulative arrival rate of inhomogeneous Poisson point processes up to any one time point, with guarantees that hold uniformly over (continuous) time. This allows the continuous monitoring of such event streams and testing hypotheses such as about differences in cumulative arrival rates at all times, so that, for example, if a difference is detected action can be taken right away.
% Our contributions are motivated by the need to continuously monitor arrival processes between arms of an A/B test, such as quickly detecting increases in the arrival rates of error logs in software experiments \citep{anytimecounts}. 
The ability to test hypotheses repeatedly through time, while maintaining strict Type I error guarantees, necessitates careful sequential analysis.

In section \ref{sec:one_process} we construct a confidence process $C^\alpha(t)$, a continuous-time analogue of the discrete-time confidence sequence \citep{darling67, howard}, for a single inhomogeneous Poisson process that is guaranteed to cover the cumulative arrival rate $\Lambda(t)=\int_0^t\lambda(s)\mathrm ds$ simultaneously at all times $t>0$ with probability at least $1-\alpha$. We further show that $C^\alpha(t)$ is powered to detect deviations in average arrival, namely that if $\Lambda(t)/t\to\bar\lambda$ then any $\Lambda_1(t)$ with $\Lambda_1(t)/t\to\bar\lambda_1\neq\bar\lambda$ is eventually excluded with probability one.
In section \ref{sec:two_processes}, we extend our methods and results to the case of two independent processes with intensity functions $\lambda^A(t),\lambda^B(t)$. We consider a confidence process taking set values on the plane that contain $(\Lambda^A(t), \Lambda^B(t))$ jointly at all times $t$ with probability at least $1-\alpha$. This allows the intensity measure, or equivalently the time-averaged instantaneous rates $(\Lambda^A(t) / t, \Lambda^B(t) / t)$ to be estimated and monitored at all times, taking appropriate actions when necessary.

In section \ref{sec:testing} we then turn to testing equality of arrival processes, that is, the null hypothesis $H_0: \lambda^A(t) = \lambda^B(t)$ for all $t > 0$, for which we provide a continuous-time $e$-process.
% To complement estimation we also provide $e$-processes for testing hypotheses in section \ref{sec:testing}, with a particular focus on the null hypothesis $H_0: \lambda^A(t) = \lambda^B(t)$ for all $t > 0$. 
That is, we construct a nonnegative process $E(t)$ such that, under the null hypothesis $H_0$, we have $\mathbb{E}_{H_0}[E(\tau)] = 1$ for any, possibly data-dependant, random stopping time $\tau$ \citep{grunwald, grunwaldramdas}. Rejecting the null at time $\tau = \inf\{t>0 : E(t) > \alpha^{-1}\}$ then has the Type-I error guarantee $\mathbb{P}_{H_0}[\tau<\infty] \leq \alpha$.
% from Markov's inequality $\mathbb{P}[\tau<\infty] \leq \mathbb{E}[E(\tau)]\alpha = \alpha$.
% The Type-I error guarantee under $H_0$ follows from Markov's inequality $\mathbb{P}[E(\tau) \geq \alpha^{-1}] \leq \mathbb{E}[E(\tau)]\alpha = \alpha$. 
Our $e$-processes are constructed directly from our joint confidence processes on $(\Lambda^A(t), \Lambda^B(t))$, providing a unified approach to inference. We also observe that our construction coincides with the universal $e$-process of \citet{Wasserman_2020} applied to our setting. We characterize the asymptotic growth rate of the proposed $e$-process when $\Lambda^A(t)/t\to\bar\lambda^A$, $\Lambda^B(t)/t\to\bar\lambda^B$, and we show the rate depends on a particular divergence between $\bar\lambda^A$ and $\bar\lambda^B$. As a consequence, this establishes that our test is powered to detect violations of $H_0$ when $\bar\lambda^A\neq\bar\lambda^B$, that is, $\mathbb P(\tau<\infty)=1$ in such cases.
We further compare the $e$-process with two alternative procedures in section \ref{sec:comparisons}: one based on the Bernoulli draw of the next event coming from $A$ or $B$ and one based on a normal approximation of the difference in counting processes.
% , demonstrating that it outperforms constructing asymptotic confidence sequences for the difference in intensities directly.

We demonstrate our methodology in section \ref{sec:application} with an application regarding detecting increases in customer call center volume among two treatment arms of the A/B test, where the desire is to abandon the experiment as quickly as possible if any increase is detected.

\subsection{Related Literature}
Many classical statistical tests fail to provide \textit{time-uniform} Type-I error and coverage guarantees, providing guarantees at a fixed sample size only, because the experiments that motivated their development, such as agriculture, yielded results simultaneously \citep{anscombe}. In modern experiments, however, outcomes are often observed sequentially. Repeated hypothesis testing with such methods on an accumulating set of data causes the cumulative Type-I error to grow \citep{armitage}. To facilitate continuous monitoring of hypotheses with the ability to perform early stopping, we need sequential analysis.

Early works date back to \citet{wald} and the sequential probability ratio test. \citet{darling67, robbins2} introduced confidence sequences for continuous monitoring of estimands. \citet{howard} provides nonparametric nonasymptotic confidence sequences under tail conditions. Confidence sequences for quantiles are provided by \citet{sequentialquantiles}. Confidence sequences to heavy tailed and robust mean estimation are available in \citet{Wang_2023, huberrobust}. \citet{anytimecounts} provide confidence sequences for multinomial parameters, while \citet{samplingreplacement} provide confidence sequences in sampling without replacement settings. \citet{linear_models} provide confidence sequences for linear model coefficients and average treatment effects in randomized experiments. Time-uniform versions of the central limit theorem via strong approximations yielding asymptotic confidence sequences can be found in \citet{asymptoticcs, bibaut2022}.

Confidence sequences are closely related to $e$-processes \citep{grunwald, grunwaldramdas}. The ability to stop early has been valuable in medical survival analysis \citep{logrank} and contingency table testing \citep{turner2021generic, turner2023safe}, but also in online A/B testing \citep{johari, johari2}. \citet{rapidregression} use confidence sequences to monitor quantiles of performance metrics in software A/B experiments. \citet{ham2022} construct design-based confidence sequences on sample causal estimands to de-risk experimental product changes.

In contrast to earlier works where typically a continuous or discrete outcome is observed, and the goal of inference is to construct a confidence sequence on the mean or average treatment effect, our setting is different in the sense that we only observe the timestamp at which an event occurs, and the goal of inference is to construct a confidence process on the instantaneous rate of events.

\section{Inhomogeneous Poisson Point Processes}
\label{sec:review}

We start by defining the inhomogeneous Poisson process.
For a countable set $\mathcal{T} \subseteq \mathbb{R}$ and a Borel set $B \in \mathcal{B}$, let $N(B) \coloneqq \# \mathcal{T} \cap B$ denote the number of points in $B$. An inhomogeneous Poisson point process over $\mathbb{R}$ is a random countable set $\mathcal{T} \in \mathbb{R}$, defined by a nonnegative locally-integrable intensity function $\lambda: \mathbb{R} \rightarrow \Rp$, which we denote $\mathcal{T} \sim \mathcal{P}(\lambda)$, with the following three properties
\begin{enumerate}
    \item $\mathbb{E}[N(B)] =\Lambda(B)= \int_{\mathcal{B}}\lambda(s)ds$ for any Borel set $B \in \mathcal{B}$,
    \item For any finite collection of disjoint sets $B_1, B_2, \dots, B_K$, 
    \begin{enumerate}
    \item the random counts $N(B_1), N(B_2), \dots, N(B_K)$ are independent,
    \item $N(B_i) \sim \text{Poisson}\left(\Lambda(B_i)\right)$.
\end{enumerate}
\end{enumerate}
$\Lambda(B)$ is named the \textit{intensity measure} and is the expected number of points belonging to the set $B$. In this work, we consider observing an inhomogeneous Poisson point process continuously from time $t_0 = 0$ and restrict our focus to intervals of the form $[0,t]$, where $t$ denotes the current time. For simplicity we use the shorthand $N(t) = \# \mathcal{T}\cap[0, t]$ and $\Lambda(t) = \int_0^t\lambda(s)ds$ to denote the counting process and intensity measure respectively over the current observation window $[0, t]$. Also, for any $\lambda_\circ^\diamond(t)$ with any subscript or superscript we implicitly also define the corresponding $\Lambda_\circ^\diamond(t) = \int_0^t\lambda_\circ^\diamond(s)ds$. We write $\mathcal{F}_t = \sigma(N(s) : 0 \leq s \leq t)$ to denote the canonical filtration.

Let $\mathcal{S} = \{[a, b]:  a,b \in \Rp, b > a\}$ denote a set of intervals in $\Rp$. Our first result is a confidence process $C^\alpha : \Rp \rightarrow \mathcal{S}$ adapted to the filtration $(\mathcal F_t)_{t\in\Rp}$, meaning that $C^\alpha(t)$ is measurable with respect to $\mathcal F_t$, such that
\begin{equation}
    \mathbb{P}\left[\Lambda(t) \in C^\alpha(t)  \,\,\, \text{for all} \,\,\, t > 0 \right] \geq 1-\alpha.
\end{equation}
Note that $C^\alpha(t)$ is random as it depends on the observed data, but we suppress this dependence for ease of notation.
Writing $C^\alpha(t) = [l(t), u(t)]$, then a second confidence process $\bar{C}^\alpha(t):=[l(t) / t, u(t)/ t]$ can be defined that covers $\frac{1}{t}\Lambda(t) = \frac{1}{t}\int_0^t \lambda(s)ds$ with the same time-uniform coverage guarantee. This enables the average intensity to be continuously estimated as the point process continues to be observed.

Our second result is a sequential test of the hypothesis $\lambda^B(t) = \lambda^A(t)$. The test is sequential in the sense that at any current time $t$, the hypothesis can be tested given the observations of $\mathcal{T}^A \cap [0,t]$ and $\mathcal{T}^B \cap [0,t]$. Moreover, the test can be applied as frequently as desired without sacrificing the desired Type-I error guarantee.
\section{Confidence Processes for the Intensity Measure $\Lambda(t)$}
\label{sec:one_process}
At time $t$, the intersection of $\mathcal{T}\cap [0,t]$ is observed by the experimenter. Given this observation, the likelihood function for $\lambda(t)$ is 
\begin{equation}
    \mathcal{L}_t(\lambda) \propto e^{-\int_0^t \lambda(s) ds}\prod_{x \in \mathcal{T}\cap [0,t]} \lambda(x) .
\end{equation}
Consider the test of a simple null $\lambda_0$ vs a simple alternative $\lambda_1$. The likelihood ratio at time $t$ is
\begin{equation}
\label{eq:likelihood_ratio}
    \frac{\mathcal{L}_t(\lambda_1)}{\mathcal{L}_t(\lambda_0)} = e^{-\int_0^t \lambda_1(s)-\lambda_0(s) ds}\prod_{x \in \mathcal{T}\cap [0,t]} \frac{\lambda_1(x)}{\lambda_0(x)} .
\end{equation}
Our first result states that the likelihood ratio in equation \eqref{eq:likelihood_ratio} is a continuous time non-negative supermartingale under the null hypothesis.
\begin{theorem}
\label{thm:simplemartingale}
    Suppose $\mathcal{T} \sim \mathcal{P}(\lambda_0)$, then for any $t >0$, $\delta > 0$, and $\lambda_1 : \mathbb{R}\rightarrow \mathbb{R}_{\geq 0}$
    \begin{equation}
        \mathbb{E}\left[ \frac{\mathcal{L}_{t+\delta}(\lambda_1)}{\mathcal{L}_{t+\delta}(\lambda_0)} \mid \mathcal{F}_t\right] = \frac{\mathcal{L}_t(\lambda_1)}{\mathcal{L}_t(\lambda_0)}.
    \end{equation}
The likelihood ratio process satisfies the following time-uniform bound
    \begin{equation}
    \label{eq:simpleville}
        \mathbb{P}\left[\sup_{t\in\Rp} \frac{\mathcal{L}_t(\lambda_1)}{\mathcal{L}_t(\lambda_0)} \geq \alpha^{-1}\right] = \alpha
    \end{equation}
\end{theorem}
The likelihood ratio process is a continuous time nonnegative supermartingale under the null hypothesis, regardless of the chosen alternative, according to theorem \ref{thm:simplemartingale}. The time-uniform bound in equation \eqref{eq:simpleville} follows immediately from Ville's inequality for continuous time nonnnegative supermartingales \citep{ville}.

Consider a proportional hazards alternative $\lambda_1(t) = e^\theta \lambda_0(t)$. The likelihood ratio can now be written as
\begin{equation}
\label{eq:likelihood_ratio_theta}
    \frac{\mathcal{L}_t(\lambda_1)}{\mathcal{L}_t(\lambda_0)} = e^{\theta N(t) - (e^{\theta}-1)\Lambda_0(t)}.
\end{equation}
The motivation for working with a proportional hazards alternative is the simplification it brings to the likelihood ratio. The data now enters the likelihood ratio only through $N(t)$, and similarly $\lambda_0$ now enters the likelihood ratio only through $\Lambda_0(t)$. The alternative in equation \eqref{eq:likelihood_ratio_theta} is simple, corresponding to a singleton $\theta$. In theorem \ref{thm:single_hypothesis_test} we will also show that this choice does not in fact limit us to detecting only such alternatives.

A test for a composite alternative can be obtained by mixing the likelihood ratio with respect to a mixture distribution on $\theta$. In practice, distributions which are conjugate to the likelihood are preferred as they yield closed-form expressions for the likelihood ratio mixture. A $\text{logGamma}(\phi, \phi)$ mixture over $\theta$, which in turn implies a $\text{Gamma}(\phi, \phi)$ mixture over $e^{\theta}$, is conjugate to the likelihood ratio in equation \eqref{eq:likelihood_ratio_theta}, resulting in a convenient closed form expression \citep{howard}. This mixture distribution has mean 1 and variance $1/\phi$. It is therefore centered on the null hypothesis and $\phi$ acts as a precision parameter to concentrate the mixture toward local alternatives.
\begin{theorem}
\label{thm:compositemartingale}
    Assume $\mathcal{T} \sim \mathcal{P}(\lambda_0)$ and for any fixed $\phi > 0$ let $\Pi$ denote a $\text{logGamma}(\phi, \phi)$ distribution. Define    
    \begin{equation}
\label{eq:composite_martingale}
            M^{\Lambda_0}(t) := \int   \frac{\mathcal{L}_t(e^\theta\lambda_0)}{\mathcal{L}_t(\lambda_0)} d\Pi(\theta)
            =M(N(t),\Lambda_0(t);\phi),
            % \frac{\phi^\phi}{(\phi + \Lambda_0(t))^{\phi+N(t)}}\frac{\Gamma(\phi+N(t))}{\Gamma(\phi)}e^{\Lambda_0(t)}.
\end{equation}
where
$$    M(n,L;\phi)=\frac{\phi^\phi}{(\phi + L)^{\phi+n}}\frac{\Gamma(\phi+n)}{\Gamma(\phi)}e^{L}.
$$
The process $M^{\Lambda_0}(t)$ is a continuous-time nonnegative supermartingale.

Moreover,
\begin{equation}
    \label{eq:simpleville2}\mathbb{P}[\exists t> 0 : M^{\Lambda_0}(t) \geq \alpha^{-1}] = \alpha.\end{equation}
\end{theorem}

Equation~\eqref{eq:simpleville2} again follow from Ville's inequality.
% From Ville's inequality $\mathbb{P}[\exists t> 0 : M^{\Lambda}(t) \geq \alpha^{-1}] = \alpha$ when $\mathcal{T} \sim \mathcal{P}(\lambda)$. 
The null hypothesis $H_0: \mathcal{T} \sim \mathcal{P}(\lambda_0)$ can therefore be rejected as soon as $M^{\Lambda_0}(t)$ becomes larger the rejection threshold of $\alpha^{-1}$, while maintaining a time-uniform Type-I error guarantee of $\alpha$. 

Since the numerator of $M^{\Lambda_0}(t)$ only mixes over proportional hazard alternatives, one might worry that the test described in the previous paragraph is only powered to detect such alternatives. The next theorem demonstrates that in fact this test has power 1 against \textit{any} alternative for which the average rate of arrivals differs from the average rate under the null hypothesis.
\begin{theorem}
\label{thm:single_hypothesis_test}
    Let $\lambda(t),\lambda_0(t)$ be given such that $\Lambda(t) / t \rightarrow \bar{\lambda}$, $\Lambda_0(t) / t \rightarrow \bar{\lambda}_0$ and $\lambda(t)$ is bounded over $t>0$.
    % Assume $\mathcal{T} \sim \mathcal{P}(\lambda)$ with $\Lambda(t) / t \rightarrow \bar{\lambda}$ and $\lambda(t)$ bounded for all $t>0$, 
    % whereas under the null hypothesis $H_0: \mathcal{T} \sim \mathcal{P}(\lambda_0)$ with  $\Lambda_0(t) / t \rightarrow \bar{\lambda}_0 \neq \bar{\lambda}$, 
    Suppose $\mathcal{T} \sim \mathcal{P}(\lambda)$.
    Then, 
    \begin{equation}
        \frac{\log M^{\Lambda_0}(t)}{t} \overset{a.s.}{\rightarrow} \bar{\lambda}\log\frac{\bar{\lambda}}{\bar{\lambda}_0}-(\bar{\lambda}-\bar{\lambda}_0)=D_{KL}(\bar{\lambda}||\bar{\lambda}_0),
    \end{equation}
    where $D_{KL}(\bar{\lambda}||\bar{\lambda}_0)$ is the Kullback-Leibler divergence of a $\text{Poisson}(\bar{\lambda})$ distribution from a $\text{Poisson}(\bar{\lambda}_0)$ distribution.
\end{theorem}
Theorem \ref{thm:single_hypothesis_test} states that whenever arrivals come from a process with a different asymptotic average arrival than the null hypothesis, then $M^{\Lambda_0}(t)$ grows exponentially quickly under the alternative, providing the exact asymptotic growth rate. Consequently, it will also almost surely grow past $\alpha^{-1}$ leading us to reject the null eventually with probability one.

By considering the complement of equation~\ref{eq:simpleville2},
% Ville's inequality, $\mathbb{P}[\forall t > 0 : M^{\Lambda}(t) < \alpha^{-1}] = 1-\alpha$, 
we obtain a confidence process for the intensity measure $\Lambda$. 
\begin{corollary}
    For any fixed $\phi > 0$, the sets defined by 
    \begin{equation}
    \label{eq:confidence}
                C^\alpha(t) =\left\{L \in \Rp : M(N(t),L;\phi)\leq \alpha^{-1}\right\}
    \end{equation}
    provide a $1-\alpha$ confidence process for $\Lambda(t)$. Namely, if $\mathcal{T} \sim \mathcal{P}(\lambda_0)$, then
    $$
    \mathbb P[\Lambda_0(t)\in C^\alpha(t)~\forall t>0]=1-\alpha.
    $$
    Moreover, if $\lambda_0(t)$ is bounded over $t>0$ and $\Lambda_0(t)/t\to\bar\lambda_0$ then
    $$
    \mathbb P[\Lambda(t)\notin C^\alpha(t)~\forall t>0] = 1
    $$
    for any $\Lambda(t)$ with $\Lambda(t)/t\to\bar\lambda\neq\bar\lambda_0$.
\end{corollary}

Note that the confidence process $C^\alpha(t)$ is particularly simple because $\Lambda_0(\cdot)$ only enters into $M^{\Lambda_0}(t)$ as $\Lambda_0(t)$, its value at $t$. Had this not been the case, we would generally need to consider $C^\alpha(t)$ to be a set of possible $\Rp\to\Rp$ \textit{functions} $\Lambda(\cdot)$, rather than just a set in $\Rp$. This significant simplification is due to our specific choice of proportional hazard alternative. 
This is important because it makes $C^\alpha(t)$ interpretable and practically usable as we can simply plot it over time (see figure \ref{fig:counting}).
Moreover, since $M(n,L;\phi)$ is convex in $L$, the set defined in equation \eqref{eq:confidence} is an interval, which can be computed numerically using a root-finding algorithm. 

While benefiting greatly in terms of both computation and interpretability of $C^\alpha(t)$ from how we chose the space of alternative arrival processes when defining $M^{\Lambda_0}(t)$, we do not pay much for this choice in terms of power. We show we are able to exclude any alternative that converges at any distinct average arrival. This is made obvious by figure \ref{fig:average} showing how our confidence process scaled by $1/t$ concentrates at a point.

\iffalse
As a good default, $\phi = 1$ gives an appealingly simple simplification
\begin{equation}
        M^{1, \Lambda_0}(t) :=  \frac{1}{(1 + \Lambda_0(t))^{1+N(t)}}N(t)!e^{\Lambda_0(t)}
\end{equation}
\fi
\begin{example}
\label{ex:single_process}
In the following simulated example $\mathcal{T} \sim \mathcal{P}(\lambda)$ with $\lambda(t) =e^{3\sin(2\pi t/20)}$. Figure \ref{fig:counting} shows the realized point process and the confidence process $C^\alpha(t)$, which we remark covers $\Lambda(t) = \mathbb{E}[N(t)]$ at all times in this particular sample path.  Figure \ref{fig:average} shows the confidence process on $(1/t)\Lambda(t)$.

\begin{figure}
    \centering
    \includegraphics[width=\linewidth]{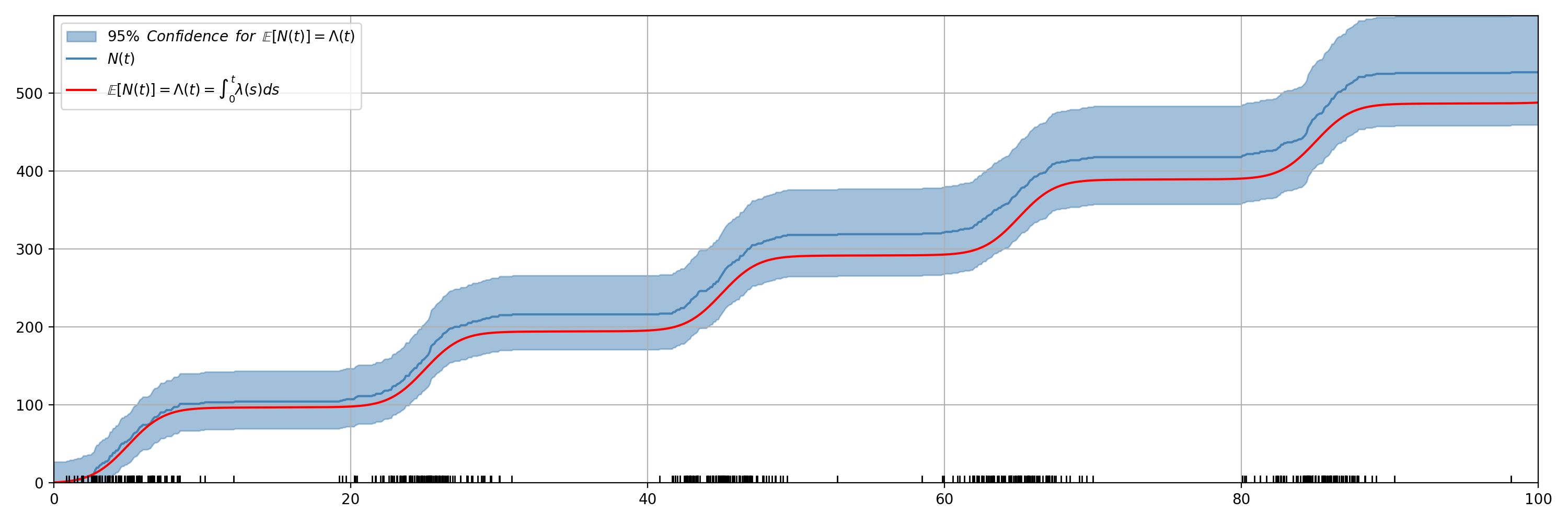}
    \caption{The intensity measure $\Lambda(t) = \mathbb{E}[N(t)]$ (red), the observed counting process $N(t)$ (blue), the $0.95$ confidence process $C^\alpha(t)$ (shaded blue), and the realized point process (black ticks).}
    \label{fig:counting}
\end{figure}
\begin{figure}
    \centering
    \includegraphics[width=\linewidth]{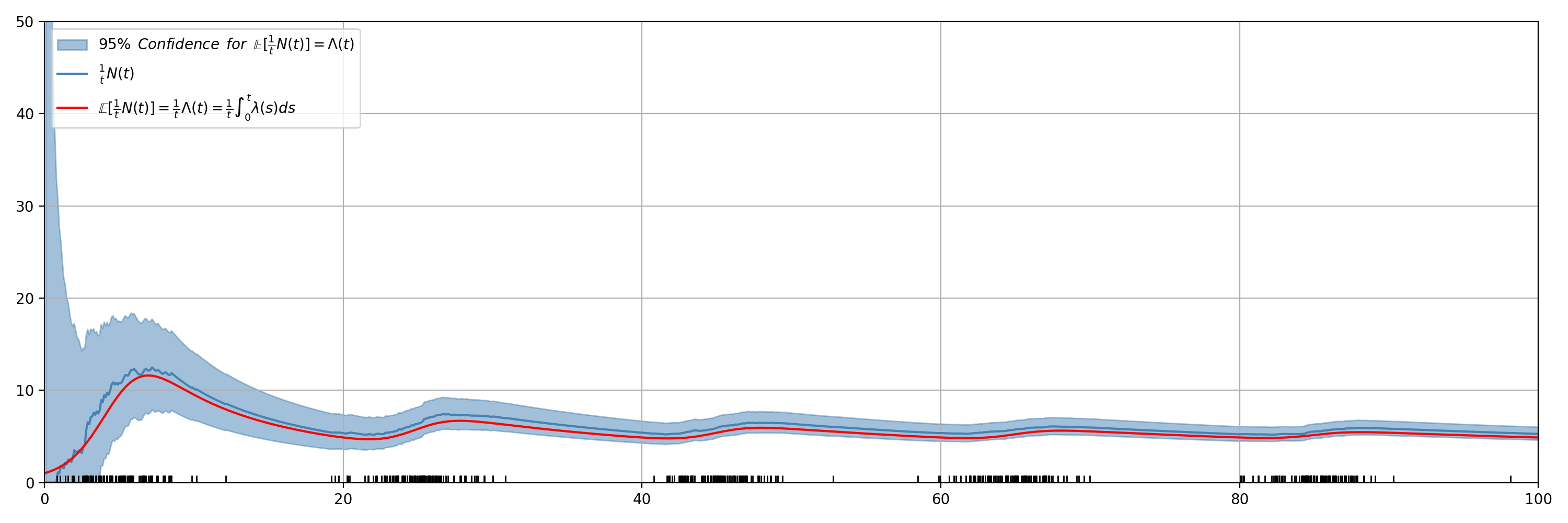}
    \caption{The time-average of intensity $\frac{1}{t}\Lambda(t)$ (red), the estimated time-average of intensity (blue), the $0.95$ confidence process (shaded blue), and the realized point process (black ticks).}
    \label{fig:average}
\end{figure}
\end{example}

\section{Confidence Processes for Dual Intensity Measures}
\label{sec:two_processes}
Now suppose the experimenter observes two independent realizations $\mathcal{T}^B \sim \mathcal{P}(\lambda^B)$ and $\mathcal{T}^A\sim \mathcal{P}(\lambda^A)$. Our second result provides a confidence process on $(\Lambda^A(t), \Lambda^B(t))$ \textit{jointly} based on observations $\mathcal{T}^B\cap [0,t]$ and $\mathcal{T}^A\cap [0,t]$ at time $t$, from which a confidence process on $\Lambda^B(t) - \Lambda^A(t)$ can be derived.

\begin{theorem}
\label{thm:joint_martingale}
    Assume $\mathcal{T}^A \sim \mathcal{P}(\lambda_0^A)$ and $\mathcal{T}^B \sim \mathcal{P}(\lambda_0^B)$ independently. For any fixed $\phi > 0$ define
    \begin{equation}
    \label{eq:joint_martingale}
        M^{\Lambda_0^A, \Lambda_0^B}(t) =  M^{\Lambda^A_0}(t)M^{\Lambda^B_0}(t)
\end{equation}
The process $M^{\Lambda_0^A, \Lambda_0^B}(t)$ is a continuous time nonnegative supermartingale.
\end{theorem}
\noindent Theorem \ref{thm:joint_martingale} follows readily from theorem \ref{thm:compositemartingale}, as a product of two independent nonnegative supermartingales is itself a nonnegative supermartingale.
\begin{corollary}
    For any fixed $\phi > 0$, the sets defined by 
    \begin{equation}
\label{eq:joint_confidence}C^\alpha(t) = \{ (L^A, L^B) \in \Rp^2 :  M(N^A(t),L^A;\phi)M(N^B(t),L^B;\phi) \leq \alpha^{-1}\}
\end{equation}
form a $1-\alpha$ confidence process for $(\Lambda^A(t), \Lambda^B(t))$. Namely, if $\mathcal{T^A} \sim \mathcal{P}(\lambda^A),\mathcal{T^B} \sim \mathcal{P}(\lambda^B)$, then
    $$
    \mathbb P[(\Lambda^A(t),\Lambda^B(t))\in C^\alpha(t)~\forall t>0]=1-\alpha.
    $$
\end{corollary}
\noindent The set $C^\alpha(t)$ is convex because $M(N^A(t),L^A;\phi)M(N^B(t),L^B;\phi)$ is a product of two univariate convex functions in $L^A$ and $L^B$, respectively. A joint confidence set on $(\Lambda^B(t) - \Lambda^A(t), \Lambda^A(t) + \Lambda^B(t))$ is obtained via the mapping $T(C^\alpha(t))$ where $T : (x,y) \rightarrow (y-x, y+x)$. These sets are visualized in figure \ref{fig:confidence_sets}.
\begin{figure}
    \centering
    \includegraphics[width=\linewidth]{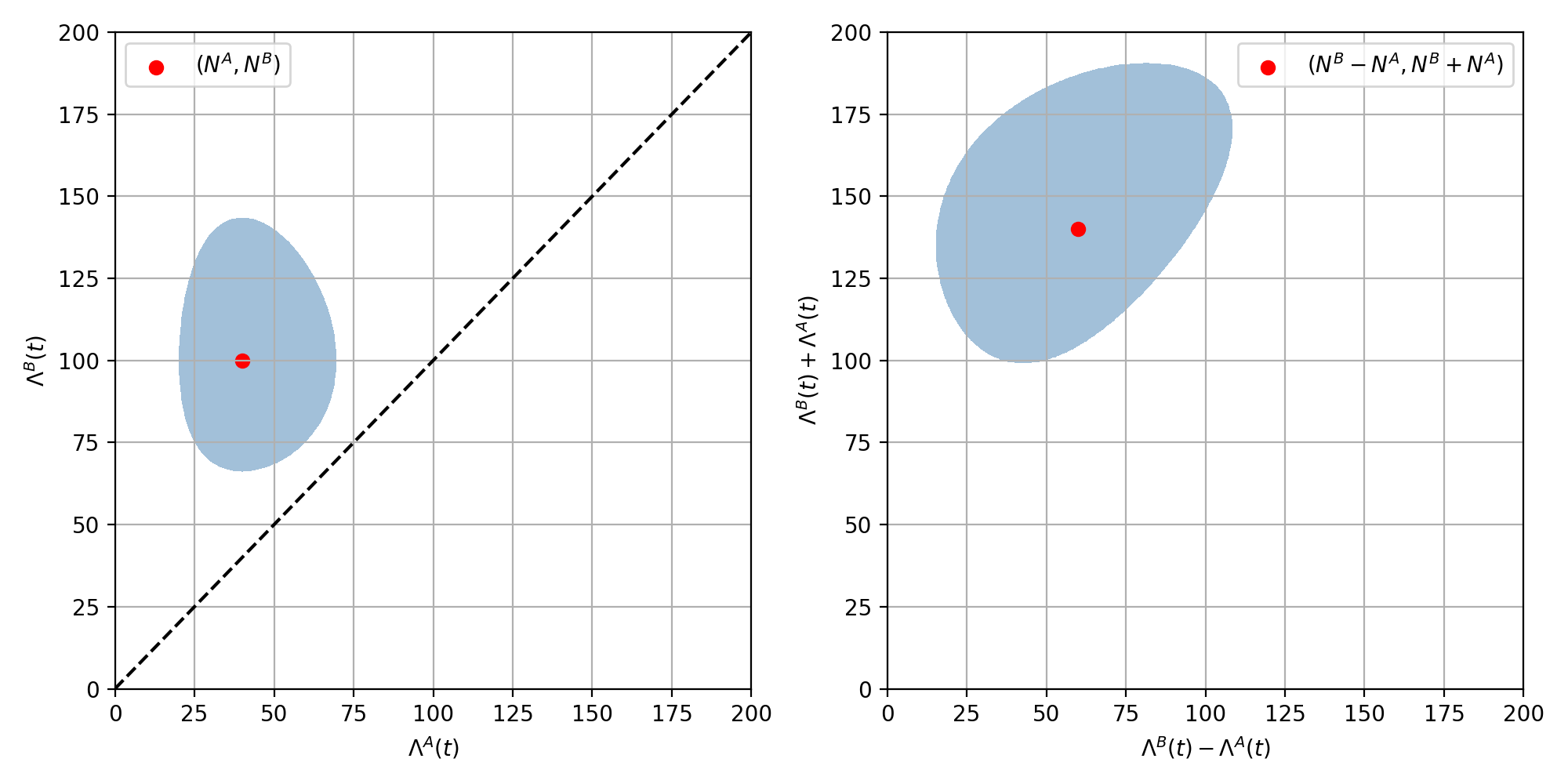}
    \caption{(Left) $1-\alpha$ Confidence set $C^\alpha(t)$ for $(\Lambda^A(t), \Lambda^B(t))$. (Right) $1-\alpha$ confidence set $T(C^\alpha(t))$, where $T : (x,y) \rightarrow (y-x, y+x)$, for $(\Lambda^B(t)-\Lambda^A(t), \Lambda^B(t)+\Lambda^A(t))$. Parameters: $\alpha=0.05$, $\phi = 1$, $N^A = 40$ and $N^B = 100$.}
    \label{fig:confidence_sets}
\end{figure}

\noindent An advantage of the joint confidence process on $(\Lambda^A(t), \Lambda^B(t))$ is that it yields simultaneous confidence processes on functionals thereof, such as $\Lambda^B(t)$, $\Lambda^A(t)$, and the difference $\Lambda^B(t) - \Lambda^A(t)$ for example. Confidence processes for $\Lambda^B(t) - \Lambda^A(t)$ can be obtained by simply projecting $T(C^\alpha(t))$ along the first dimension, i.e., $\{x \in \mathbb{R}_{\geq 0} : (x,y) \in T(C^\alpha(t)) \text{ for some } y\}$. Confidence processes for $\Lambda^B(t)$ and $\Lambda^A(t)$ can be obtained similarly from $C^\alpha(t)$. The lower and upper bounds which define these intervals are trivial to find numerically, requiring only univariate root-finding algorithms, details of which are provided in the appendix.

\begin{example}
\label{ex:two_processes}
    In the following simulated example $\mathcal{T}^A \sim \mathcal{P}(\lambda^A)$ and $\mathcal{T}^B\sim\mathcal{P}(\lambda^B)$ with $\lambda^A(t) =e^{3\sin(2\pi t/20)}$ and $\lambda^B(t) =e^{2\sin(2\pi t/20)}$ respectively. Figure \ref{fig:counting_processes} shows the simultaneous $1-\alpha$ confidence processes on $\Lambda^A(t)$ and $\Lambda^B(t)$ from equation \ref{eq:joint_confidence}. Figre \ref{fig:rates} shows the simultaneous $1-\alpha$ confidence processes on $(1/t)\Lambda^A(t)$ and $(1/t)\Lambda^B(t)$.  Figure \ref{fig:difference_rates} visualizes the confidence process for $\Lambda^B(t)/t - \Lambda^A(t)/t$.

\begin{figure}
    \centering
    \includegraphics[width=\linewidth]{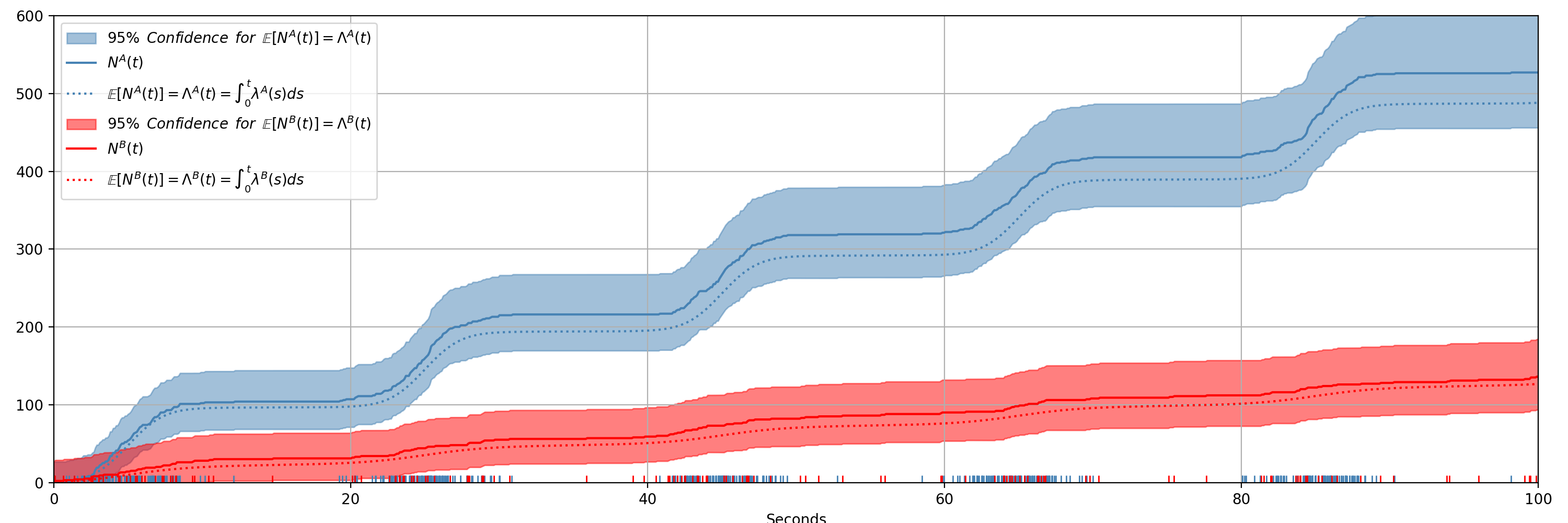}
    \caption{$1-\alpha$ simultaneous confidence processes on $\Lambda^B(t)$ and $\Lambda^A(t)$}
    \label{fig:counting_processes}
\end{figure}

\begin{figure}
    \centering
    \includegraphics[width=\linewidth]{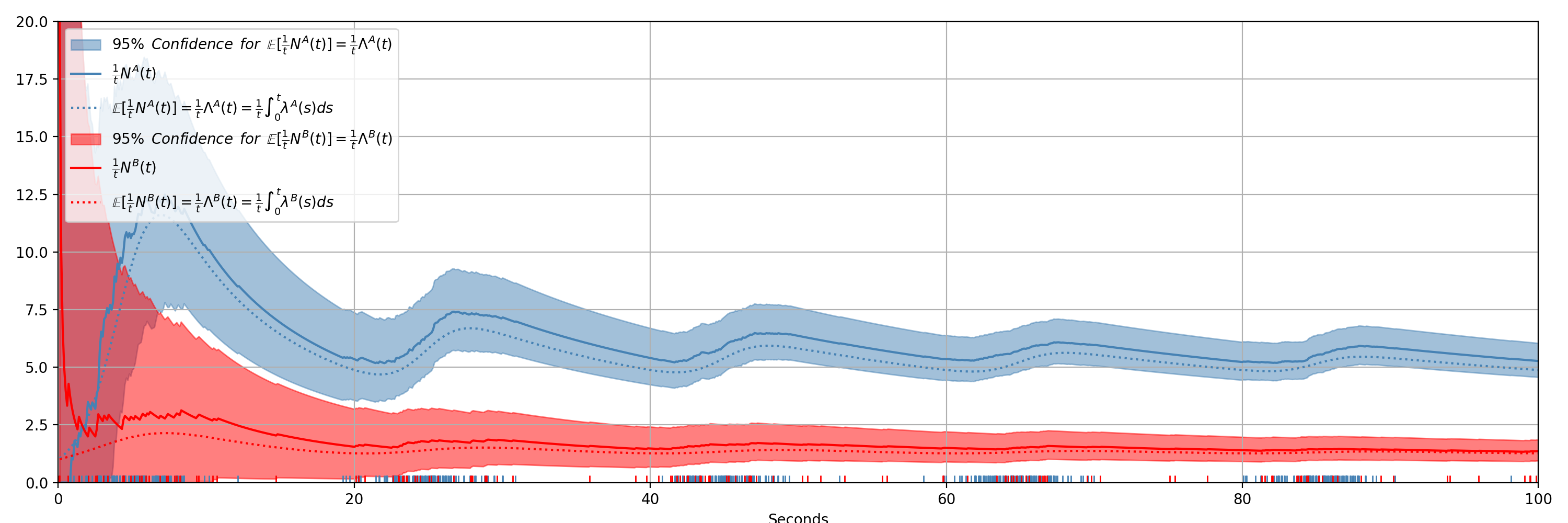}
    \caption{$1-\alpha$ confidence processes on $\Lambda^B(t)/t$ and $\Lambda^A(t)/t$}
    \label{fig:rates}
\end{figure}

\begin{figure}
    \centering
    \includegraphics[width=\linewidth]{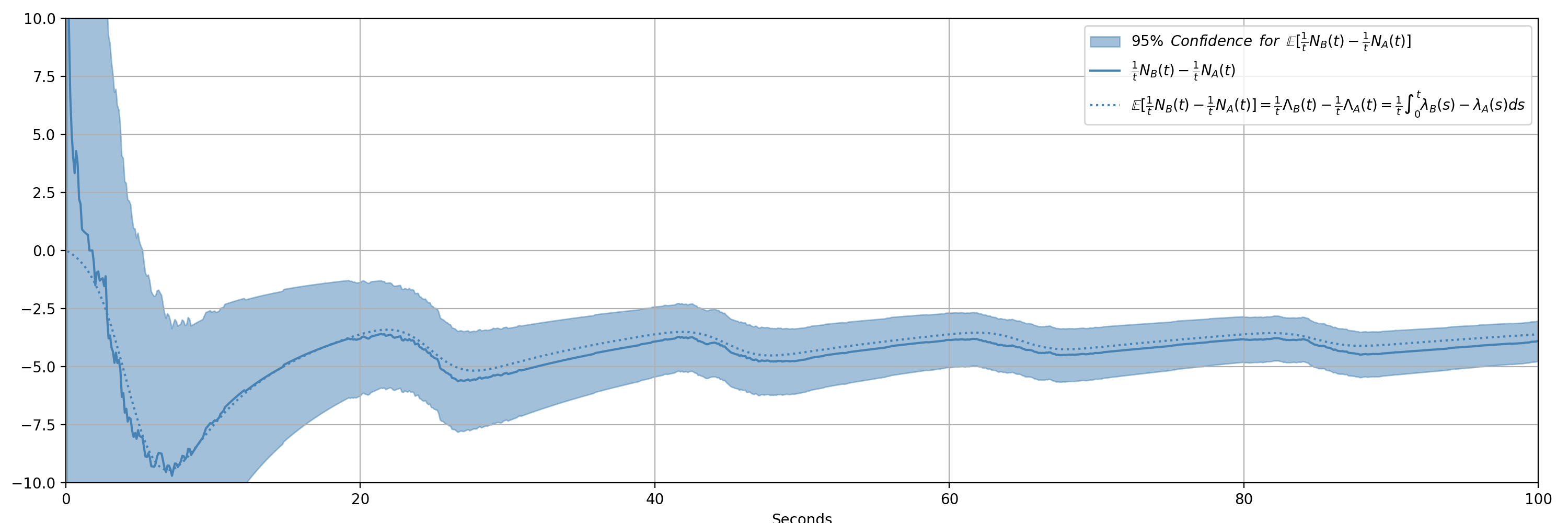}
    \caption{$1-\alpha$ simultaneous confidence processes on $\Lambda^B(t)/t-\Lambda^A(t)/t$}
    \label{fig:difference_rates}
\end{figure}

\end{example}

\section{$e$- and sequential $p$-processes for hypothesis testing}
\label{sec:testing}
We now turn to the question of testing hypotheses about the $A$ and $B$ arrival processes and in particular equality in their inhomogeneous rates.

Let $\mathcal{I}$ denote the set of all pairs of locally integrable nonnegative functions $\lambda: \mathbb{R}_{\geq 0} \rightarrow \mathbb{R}_{\geq 0}$, that is, $(\lambda^A, \lambda^B) \in \mathcal{I}$. A null hypothesis $H_0$ concerning $\lambda^A$ and $\lambda^B$ defines the subset $\mathcal{I}_0 \subset \mathcal{I}$. Let $\Theta_0(t) = \{(x,y) \in \mathbb{R}^2_{\geq 0} : x = \int_0^t f(s)ds, y=\int_0^t g(s)ds, (f,g) \in \mathcal{I}_0\}$ denote the set of values taken by the intensity measures at time $t$ which are consistent with the null hypothesis. The hypothesis $H_0$ can be rejected at the $\alpha$-level at the earliest $t$ for which $C^\alpha(t) \cap \Theta_0(t) = \emptyset$. For example, the null hypothesis of equality $\lambda^A(t)=\lambda^B(t)$ $\forall t$ implies $\Lambda^A(t) = \Lambda^B(t)$ $\forall t$ and can be rejected at the $\alpha$ level as soon as $C^\alpha(t) \cap \{(x,x) : x \in \mathbb{R}_{\geq 0}\} = \emptyset$. 

To measure the extent to which the null is violated, one can determine the smallest $\alpha$-level sequential test that would have been rejected. This amounts to, at each $t$, increasing the size of the $1-\alpha$ confidence set, by decreasing $\alpha$, until it intersects with $\Theta_0$. This defines a sequential $p$-value process, providing the guarantee $\mathbb{P}\left[\underset{t>0}{\inf}\, p(t) \leq \alpha \right] \leq \alpha$, where $p(t) = \inf \{\alpha : C^\alpha(t) \cap \Theta_0  = \emptyset\}$. Equivalently, $E(t) = 1/p(t)$ defines a valid $e$-process by proposition 12 of \citet{ramdas2022admissible}.

\begin{theorem}
\label{thm:e_process}
    Assume $\mathcal{T}^A \sim \mathcal{P}(\lambda)$ and $\mathcal{T}^B \sim \mathcal{P}(\lambda) $ independently. For any fixed $\phi > 0$
    \begin{equation}
        E(t):= \prod_{i \in \{A, B\}} \frac{\phi^\phi}{(\phi + \hat{\Lambda}(t))^{\phi+N^i(t)}}\frac{\Gamma(\phi+N^i(t))}{\Gamma(\phi)}e^{\hat{\Lambda}(t)},
    \end{equation}
    where $\hat{\Lambda}(t) = \frac{1}{2}\left(N^A(t) + N^B(t)\right)$, is an $e$-process.
\end{theorem}

We observe that this approach, which defines a $1-\alpha$ multivariate confidence sequence, decreases $\alpha$ until $C^\alpha \cap \Theta_0 \neq \emptyset$ to define a sequential $p$-value $p(t)$, and sets $E(t) = 1/p(t)$, recovers the universal inference $e$-process of \citet{Wasserman_2020}. The next theorem demonstrates that for a broad class of alternatives this test is power 1.

\begin{theorem}
\label{thm:asymp_growth}
     Assume $\mathcal{T}^A \sim \mathcal{P}(\lambda^A)$ and $\mathcal{T}^B \sim \mathcal{P}(\lambda^B)$ independently, $\lambda^A(t)$ and $\lambda^B(t)$ bounded for all $t>0$, with $\frac{\Lambda^A(t)}{t}\rightarrow \bar{\lambda}^A$ and $\frac{\Lambda^B(t)}{t}\rightarrow \bar{\lambda}^B$ as $t\rightarrow \infty$. Let $\bar{\lambda}^M = \frac{1}{2}\left(\bar{\lambda}^B+\bar{\lambda}^A\right)$, then for any $\phi > 0$
     \begin{equation}
     \label{eq:asymp_growth}
     \begin{split}
                         &\frac{\log E(t)}{t}\overset{a.s.}{\rightarrow} D_{KL}(\bar{\lambda}^B||\bar{\lambda}^M)+D_{KL}(\bar{\lambda}^A||\bar{\lambda}^M)\\
                         &\qquad=\bar{\lambda}^A\log \frac{2\bar{\lambda}^A}{\bar{\lambda}^A+\bar{\lambda}^B} + \bar{\lambda}^B\log \frac{2\bar{\lambda}^B}{\bar{\lambda}^A+\bar{\lambda}^B}\hspace{1cm} 
     \end{split}
     \end{equation}
     where $D_{KL}(\bar{\lambda}^B||\bar{\lambda}^M)$ is the Kullback-Leibler divergence of a $\text{Poisson}(\bar{\lambda}^B)$ distribution from $\text{Poisson}(\bar{\lambda}^M)$
\end{theorem}
\noindent Note the term on the right hand side is zero when $\bar{\lambda}^A = \bar{\lambda}^B$ and positive otherwise. 
% In the next section we compare this asymptotic growth rate to two other methods
As a consequence, when the assumptions of theorem \ref{thm:asymp_growth} hold with $\bar{\lambda}^A \neq \bar{\lambda}^B$ then our test will eventually reject the equality hypothesis $\lambda^A=\lambda^B$ with probability one:
$$
\mathbb P[\exists t>0:E(t)\geq\alpha^{-1}]=1.
$$
And, how quickly we reject depends on the divergence between $\bar{\lambda}^A$ and $\bar{\lambda}^B$ as defined in equation \ref{eq:asymp_growth} (which looks almost like a Jenson-Shannon divergence except that $\text{Poisson}(\bar{\lambda}^M)$ is not the mixture of $\text{Poisson}(\bar{\lambda}^A)$ and $\text{Poisson}(\bar{\lambda}^B)$).

\subsection{Comparisons to Other Methods}
\label{sec:comparisons}

We now compare to two other possible alternative methods.

\subsubsection{\citet{anytimecounts}}

\citet{anytimecounts} also proposed a sequential test for $\lambda^A(t)=\lambda^B(t)$. Their observation was that, at any point in time, notwithstanding inhomogeneity, if the equality hypothesis holds then the probability the next event arrives from process $B$ is $1/2$ due to the memoryless property of the Poisson process. This simplified the problem to observing a discrete sequence of independent Bernoulli's with probability $1/2$ under the null hypothesis. They proposed an $e$-process obtained by integrating the binomial likelihood ratio with respect to a conjugate $Beta(\alpha,\beta)$ mixture. We can extend their result to continuous time by writing an $e$-process as follows
\begin{equation}
    \tilde{E}(t) = \frac{B(N^A(t)+\beta, N^B(t)+\alpha)}{B(\alpha,\beta)}2^{N^A(t)+N^B(t)},
\end{equation}
where $B(\alpha,\beta) = \Gamma(\alpha)\Gamma(\beta)/\Gamma(\alpha+\beta)$. 
We can further extend their results by characterising the asymptotic growth rate of this $e$-process.

\begin{theorem}
\label{thm:asymp_growth2}
    Under the assumptions of theorem \ref{thm:asymp_growth}, for any $\alpha,\beta > 0$
     \begin{equation}
     \label{eq:asymp_growth2}
     \begin{split}
                         &\frac{\log \tilde{E}(t)}{t}\overset{a.s.}{\rightarrow} D_{KL}(\bar{\lambda}^B||\bar{\lambda}^M)+D_{KL}(\bar{\lambda}^A||\bar{\lambda}^M)\\
                         &=\bar{\lambda}^A\log \frac{2\bar{\lambda}^A}{\bar{\lambda}^A+\bar{\lambda}^B} + \bar{\lambda}^B\log \frac{2\bar{\lambda}^B}{\bar{\lambda}^A+\bar{\lambda}^B}\hspace{1cm} 
     \end{split}
     \end{equation}
     where $D_{KL}(\bar{\lambda}^B||\bar{\lambda}^M)$ is the Kullback-Leibler divergence of a $\text{Poisson}(\bar{\lambda}^B)$ distribution from $\text{Poisson}(\bar{\lambda}^M)$
\end{theorem}
This asymptotic growth rate is identical to the rate in theorem \ref{thm:asymp_growth}. The advantage of our proposal, however, is that while it enjoys the same asymptotic growth rate for testing the equality hypothesis, it additionally also provides \textit{simultaneous} confidence processes on $\Lambda^A(t)$, $\Lambda^B(t)$ and $\Lambda^B(t)-\Lambda^A(t)$. Moreover, these are interpretable and visualizable.

\subsubsection{\citet{asymptoticcs}}
If the only goal of inference is to test the null $\Lambda^B(t)-\Lambda^A(t) = 0$ for all $t>0$, then the reader may wonder if our test, obtained by way of constructing a confidence process on $(\Lambda^A(t), \Lambda^B(t))$ first, may be suboptimal in comparison to a method which targets $\Lambda^B(t)-\Lambda^A(t)$ directly. In this section, we consider strong approximations to difference in counts $N^B(t) - N^A(t)$, approximating it as (scaled) Brownian motion with drift and leveraging the asymptotic confidence sequences of \citet{asymptoticcs}. These asymptotic confidence sequences provide a sequential analogue of fixed sample size tests based on asymptotic normality via the central limit theorem. In this section we demonstrate that this strategy actually leads to a \textit{slower} asymptotic growth compared to our proposal. This is on top of this strategy also only being approximate. For example, in the discrete-time setting, such confidence sequences are only guaranteed to provide $\alpha$-coverage in an asymptotic regime where we wait longer and longer before we even start testing \citep{bibaut2022}.

Consider discretizing time into unit intervals and consider the sequence of Poisson random variables $z^B_i = N^B(i) - N^B(i-1)$ and $z^A_i = N^A(i) - N^A(i-1)$ with means $\Lambda^B_i \coloneqq \int_{i-1}^i \lambda^B(s)ds$ and $\Lambda^A_i \coloneqq \int_{i-1}^i \lambda^B(s)ds$, respectively. One could then define the sequence of differences $Y_i \coloneqq z^B_i - z^A_i$,  noting the partial sum process $\sum_{i=1}^t Y_i = N^B(t) - N^A(t)$ and cumulative variance process $\sum_{i=1}^t \Lambda^B_i + \Lambda^A_i = \Lambda^B(t) + \Lambda^A(t)$, and apply the asymptotic confidence sequences of \citet{asymptoticcs}. To apply their result, a few regularity conditions are required. For simplicity we assume that $\lambda^B(t)$ and $\lambda^A(t)$ are bounded and that $\Lambda^B(t) / t \rightarrow \bar{\lambda}^B$ and $\Lambda^A(t) / t \rightarrow \bar{\lambda}^A$. These assumptions ensure that the cumulative variance process diverges and that consistent variance estimation is achieved in the sense $((1/t)\sum_i^t z^B_i + z^A_i) / ( (1/t)\sum_i^t \Lambda^B_i + \Lambda^A_i) = \frac{N^B(t)+N^A(t)}{\Lambda^B(t) + \Lambda^A(t)}\overset{a.s.}{\rightarrow} 1$. This yields an (asymptotic) $e$-process, for any mixture precision $\phi > 0$, given by
\begin{equation}
\label{eq:asympmsprt}
    E^A(t) = \sqrt{\frac{\phi}{\phi + t}}e^{\frac{1}{2}\frac{t}{t+\phi}Z(t)^2}
\end{equation}
where $Z(t) = \frac{N^B(t)-N^A(t)}{\sqrt{N^B(t)+N^A(t)}}$. The expression in equation \eqref{eq:asympmsprt} takes the form of a Gaussian-mixture sequential probability ratio test statistic. Our next result shows that the asymptotic growth rate of this $e$-process is actually \textit{slower} than the $e$-processes of the preceding sections.

\begin{theorem}
\label{thm:asymp_growth3}
    Under the assumptions of theorem \ref{thm:asymp_growth}, for any $\phi > 0$
     \begin{equation}
     \label{eq:asymp_growth3}
                         \underset{t\rightarrow\infty}{\lim}\frac{\log E^A(t)}{t}=\frac{1}{2}\frac{(\bar{\lambda}^B - \bar{\lambda}^A)^2}{\bar{\lambda}^B + \bar{\lambda}^A}\hspace{1cm} {a.s.}
     \end{equation}
     which is the Kullback-Leibler divergence of a $N(\bar{\lambda}^B - \bar{\lambda}^A, \bar{\lambda}^B + \bar{\lambda}^A)$ distribution from a $N(0, \bar{\lambda}^B + \bar{\lambda}^A)$ distribution.
\end{theorem}
\noindent This limit is in fact \textit{slower} than our rate in theorem \ref{thm:asymp_growth}. Figure \ref{fig:limit_ratio} shows the relative magnitude of these two limits. When the difference between $\bar{\lambda}^A$ and $\bar{\lambda}^B$ is large, the asymptotic growth rate of $E(t)$ is considerably faster than $E^A(t)$. Figure \ref{fig:e_process_realizations} shows 20 realizations $\frac{\log E^A(t)}{t}$ and $\frac{\log E(t)}{t}$ and the theoretical limits from equations \eqref{eq:asymp_growth} and \eqref{eq:asymp_growth3}. The slower rate should not be interpreted as a shortcoming of asymptotic confidence sequences, rather, it should be attributed to considering only a single sequence of differences when two independent sequences are available. Asymptotic procedures are, nonetheless, approximate and are therefore at risk of incorrectly rejecting the null hypothesis for small $t$, whereas our proposal is nonasymptotic and exactly valid over all $t>0$.

\begin{figure}
    \centering
    \includegraphics[width=0.5\linewidth]{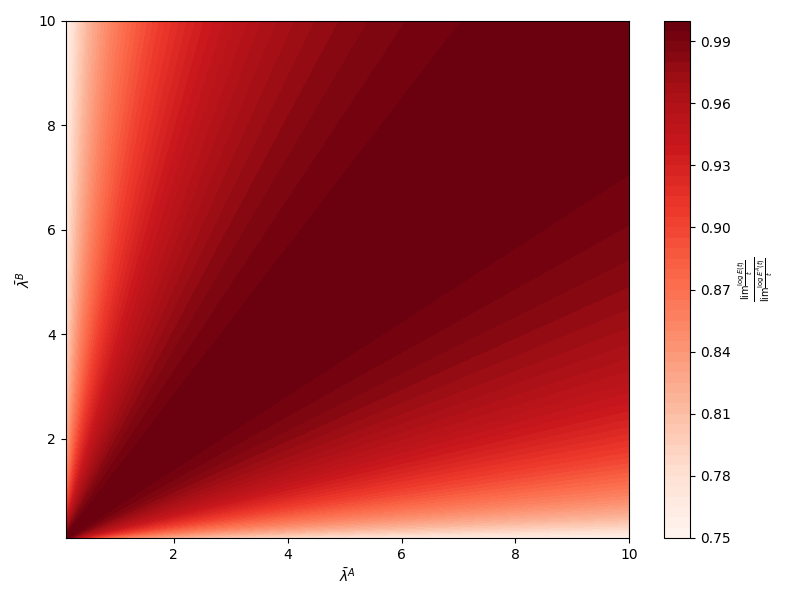}
    \caption{The value of $\underset{t\rightarrow\infty}{\lim}\frac{\log E^A(t)}{t} / \underset{t\rightarrow\infty}{\lim}\frac{\log E(t)}{t}$ for various $\bar{\lambda}^A$ and $\bar{\lambda}^B$}
    \label{fig:limit_ratio}
\end{figure}

\begin{figure}
    \centering
    \includegraphics[width=\linewidth]{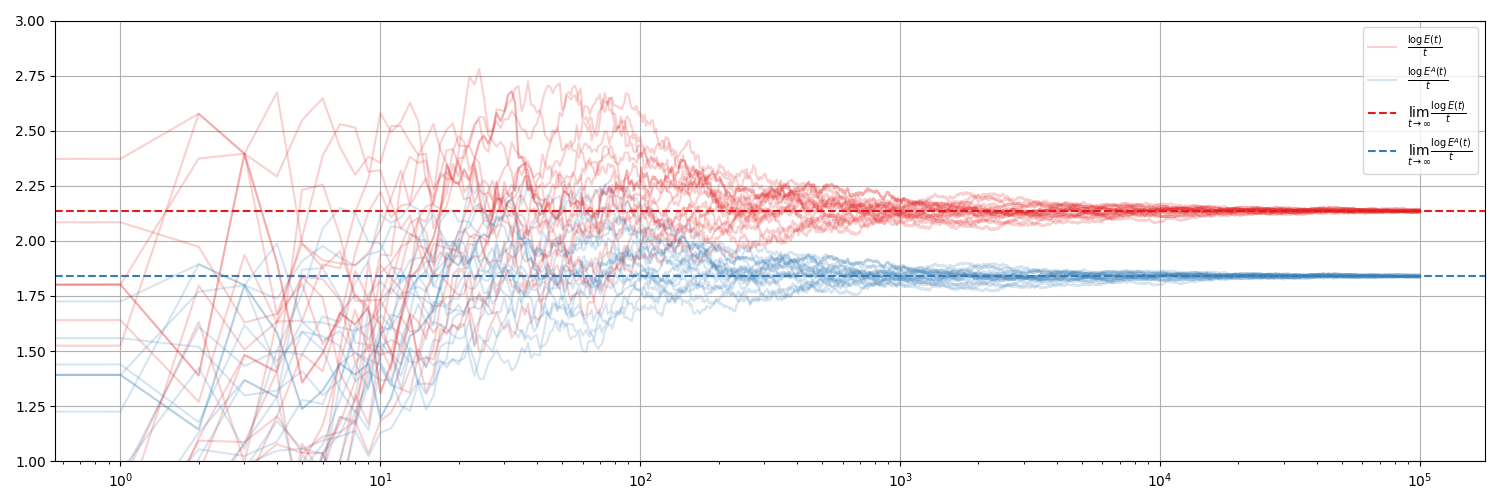}
    \caption{20 Realizations of $\frac{\log E^A(t)}{t}$ (blue) and $\frac{\log E(t)}{t}$ (red) and their associated limits from theorems \ref{thm:asymp_growth} and \ref{thm:asymp_growth3}. Simulated with $\lambda^A(t) = 0.1\lambda^B(t)$ and $\lambda^B(t) = 5 + \sin(2\pi t)$.}
    \label{fig:e_process_realizations}
\end{figure}

\section{Application: Software A/B Test}
\label{sec:application}
The following dataset is taken from an engineering A/B test on Android phone and tablet devices, provided with consent from a large internet streaming company. It was hypothesized that a specific code rewrite could increase playback speed and reduce the number of errors produced. Unfortunately, the new version of the code deployed to devices in the treatment group contained an unintended bug that in fact caused playback to crash. Figure \ref{fig:customer_rugplot} shows a dramatic increase in customer support calls classified as streaming problems.  Figure \ref{fig:customer_pvals} shows the simultaneous $0.95$ confidence process on $\Lambda^B(t) - \Lambda^A(t)$ in addition to the sequential $p$-value, where a difference is almost immediately detected at the $\alpha=0.05$ level. The raw timestamps have been scrubbed from figures to protect sensitive information. The simultaneous confidence process on the individual $\Lambda^B(t)$ and $\Lambda^A(t)$ are given in figure \ref{fig:customer_intensities}.

\begin{figure}
    \centering
    \includegraphics[width=\linewidth]{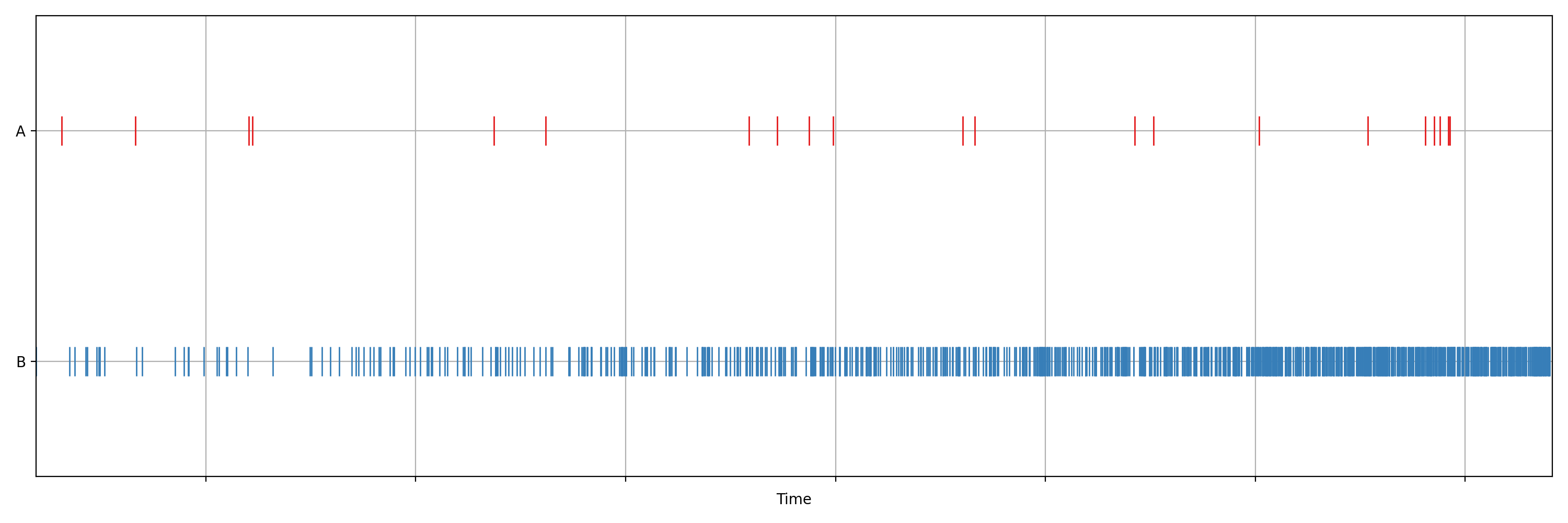}
    \caption{Arrival timestamps of customer support calls for each treatment group}
    \label{fig:customer_rugplot}
\end{figure}

\begin{figure}
    \centering
    \includegraphics[width=\linewidth]{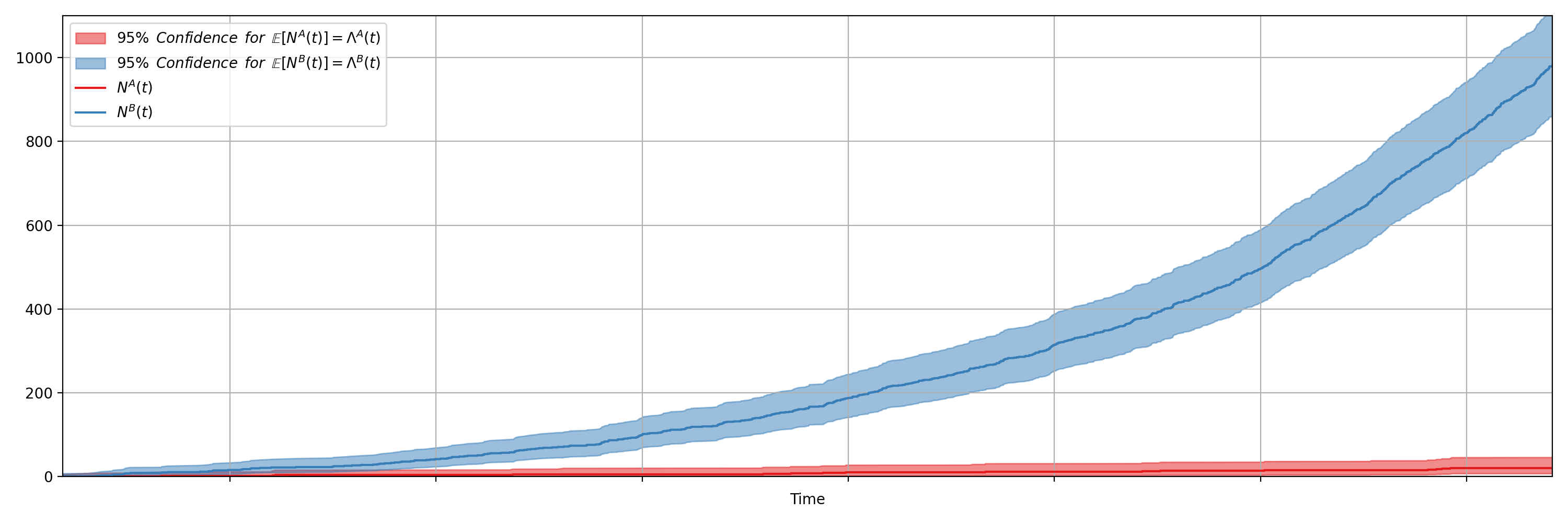}
    \caption{Counting and confidence processes on $\Lambda^B(t)$ and $\Lambda^A(t)$.}
    \label{fig:customer_intensities}
\end{figure}

\begin{figure}
    \centering
    \includegraphics[width=\linewidth]{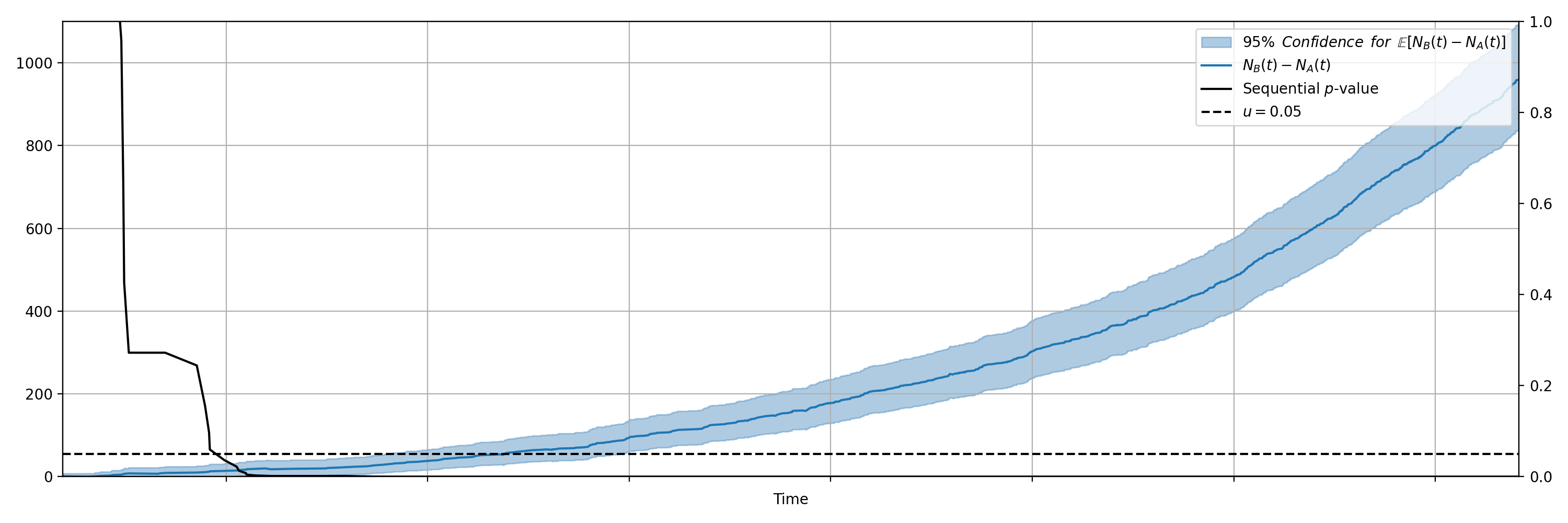}
    \caption{Confidence process on $\Lambda^B(t) - \Lambda^A(t)$ (left y-axis) and sequential $p$-value (right y-axis)}
    \label{fig:customer_pvals}
\end{figure}

\section{Conclusion}
Motivated by the need to continuously monitor arrival processes across treatment arms of a randomized experiment, we provided an anytime valid inference framework for inhomogeneous Poisson point processes. This framework provides confidence sequences on $(\Lambda^A(t), \Lambda^B(t))$, which in turn provide \textit{simultaneous} confidence sequences on $\Lambda^A(t)$, $\Lambda^B(t)$ and $\Lambda^B(t)-\Lambda^A(t)$. These confidence sequences can be used to construct $e$-processes for testing hypotheses about the intensity measures, a construction which we show is equivalent to the universal inference $e$-process. We demonstrate that this test is power 1 for a broad class of alternatives, and that the asymptotic growth rate of the $e$-process outperforms nonparametric asymptotic procedures which construct confidence sequences on the differences directly. This methodology has been successfully deployed at a large internet streaming company where it is used to monitor all active A/B tests, which we illustrate using an example where this methodology detected a increase in customer call center volume.

\bibliographystyle{agsm}
\bibliography{main}
\appendix

\section{Proofs}
\subsection{Useful lemmas}
\begin{lemma}
\label{lem:stirling}
    Stirling's Approximation
    \begin{equation}
        \log\Gamma(x) = \left(x - \frac{1}{2}\right)\log(x) - x - 2\pi + o(1),
    \end{equation}
    as $x\rightarrow \infty$.
\end{lemma}

\begin{lemma} Campbell's Theorem
\label{lemma:campbell}
\newline
Let $\mathcal{T} \sim \mathcal{P}(\lambda)$ and $R = \sum_{x \in \mathcal{T}} f(x)$ 
for some measurable function $f$ satisfying $$\int \min(|f(x)|,1)\lambda(x)dx < \infty,$$ then for complex $\theta \in \mathbb{C}$ 
\begin{equation}
    \mathbb{E}\left[e^{\theta R}\right] = e^{\int (e^{\theta f(x)} - 1)\lambda(x)dx}
\end{equation}
\end{lemma}

\begin{lemma} Strong Convergence of $\frac{1}{t}N(t)$
\label{lemma:stronglaw}
\newline
    Suppose $\frac{1}{t}\int_0^t\lambda(s) ds \rightarrow \bar{\lambda}$ with $\lambda(t) \leq \lambda_{\max}$ for all $t>0$, then $\frac{1}{t}N(t)\rightarrow \bar{\lambda}$ almost surely.
\end{lemma}
\begin{proof}
    Consider discretizing time into unit intervals, with $N_i = N(i)-N(i-1)$ and $\Lambda_i = \int_{i-1}^i \lambda(s) ds \leq \lambda_{\max}$. Note $\frac{1}{n}\sum_{i=1}^n N_i = N(n)$, similarly $\frac{1}{n}\sum_{i=1}^n \Lambda_i = \frac{1}{n}\int_0^n \lambda(s)ds$. The variances are summable in the sense
    \begin{equation*}
        \sum_{i=1}^\infty \frac{\mathbb{V}[N_i]}{i^2} = \sum_{i=1}^\infty \frac{\Lambda_i}{i^2} \leq \lambda_{\max}\frac{\pi^2}{6} < \infty
    \end{equation*}
    and so by Kolmogorov's strong law for independent random variables $\frac{1}{n}\sum_{i=1}^n N_i - \frac{1}{n}\sum_{i=1}^n \Lambda_i \overset{a.s.}{\rightarrow} 0$ $\Rightarrow$ $\frac{1}{n}\sum_{i=1}^n N_i  \overset{a.s.}{\rightarrow} \bar{\lambda}$. This demonstrates that $\frac{1}{n}N(n)\rightarrow \bar{\lambda}$. What remains is to show that $N(t)$ is controlled for $t$ between integer values. Note $N(t)/t$ is sandwiched between 
    \begin{equation*}
        \frac{N(\lfloor t \rfloor)}{t} \leq \frac{N(t)}{t} \leq \frac{N(\lceil t \rceil)}{t}.
    \end{equation*}
    The upper bound $\frac{N(\lceil t \rceil)}{t} = \frac{N(\lceil t \rceil)}{\lceil t \rceil}\frac{\lceil t \rceil}{t} \overset{a.s.}{\rightarrow} \bar{\lambda}$ and similarly the lower bound $\frac{N(\lfloor t \rfloor)}{t} = \frac{N(\lfloor t \rfloor)}{\lfloor t \rfloor}\frac{\lfloor t \rfloor}{t}\overset{a.s.}{\rightarrow} \bar{\lambda}$.
\end{proof}

\subsection{Proofs for section \ref{sec:one_process}}

\subsubsection{Proof of Theorem \ref{thm:simplemartingale}}
\label{proof:simplemartingale}
\begin{proof}
    For any time $t$ and $\delta >0$, the likelihood can be decomposed into a product of two likelihoods contributed from the observations over $[0,t]$ and $[t, t+\delta]$.
    \begin{equation*}
\begin{split}
    \frac{\mathcal{L}_{t+\delta}(\lambda)}{\mathcal{L}_{t+\delta}(\lambda_0)} &= \prod_{x \in \mathcal{T} \cap [0, t+\delta]}\frac{\lambda_1(x)}{\lambda_0(x)}e^{-\int_0^{t+\delta}\lambda_1(s) - \lambda_0(s)ds}\\
    &=\prod_{\mathcal{T} \cap [t, t+\delta]}\frac{\lambda_1(x)}{\lambda_0(x)}e^{-\int_t^{t+\delta}\lambda_1(s) - \lambda_0(s)ds} \prod_{\mathcal{T} \cap [0, t]}\frac{\lambda_1(x)}{\lambda_0(x)}e^{-\int_0^t\lambda_1(s) - \lambda_0(s)ds}\\
    &= \prod_{\mathcal{T} \cap [t, t+\delta]}\frac{\lambda_1(x)}{\lambda_0(x)}e^{-\int_t^{t+\delta}\lambda_1(s) - \lambda_0(s)ds} \frac{\mathcal{L}_t(\lambda)}{\mathcal{L}_t(\lambda_0)}.\\
\end{split}
\end{equation*}
The likelihood ratio at time $t$ is a constant with respect to the conditional expectation given $\mathcal{F}_t$.
The goal is to show the term that multiplies the likelihood ratio at time $t$ has expectation 1 conditional on the filtration $\mathcal{F}_t$.
\begin{equation}
    \begin{split}
        \mathbb{E}\left[  \prod_{\mathcal{T} \cap [t, t+\delta]}\frac{\lambda_1(x)}{\lambda_0(x)}e^{-\int_t^{t+\delta}\lambda_1(s) - \lambda_0(s)ds}  \mid \mathcal{F}_t\right] &=   \mathbb{E}\left[  \prod_{\mathcal{T} \cap [t, t+\delta]}\frac{\lambda_1(x)}{\lambda_0(x)}  \mid \mathcal{F}_t\right]  e^{-\int_t^{t+\delta}\lambda_1(s) - \lambda_0(s)ds}\\
        &=   \mathbb{E}\left[  \prod_{\mathcal{T} \cap [t, t+\delta]}\frac{\lambda_1(x)}{\lambda_0(x)}\mid \mathcal{F}_t\right]  e^{-\int_t^{t+\delta}\lambda_1(s) - \lambda_0(s)ds},\\
    \end{split}
\end{equation}
where the second line follows from independence of the process over $[0,t]$ and $[t, t+\delta]$. The expectation in the second line can be computed using lemma \ref{lemma:campbell} with $\theta = 1$ and $f(x) = e^{\log \frac{\lambda_1(x)}{\lambda_0(x)}}$. 

\begin{equation*}
\begin{split}
            \mathbb{E}\left[  \prod_{\mathcal{T} \cap [t, t+\delta]}\frac{\lambda_1(x)}{\lambda_0(x)}\mid \mathcal{F}_t \right] &=  e^{\int_t^{t+\delta}\left(\frac{\lambda_1(s)}{\lambda_0(s)} - 1\right)\lambda_0(s)ds}\\
            &=  e^{\int_t^{t+\delta}\lambda_1(s)-\lambda_0(s)ds}\\
\end{split}
\end{equation*}
Hence
\begin{equation*}
    \begin{split}
\mathbb{E}\left[ \frac{\mathcal{L}_{t+\delta}(\lambda)}{\mathcal{L}_{t+\delta}(\lambda_0)} \mid \mathcal{F}_t\right] &= \mathbb{E}\left[\prod_{\mathcal{T} \cap [t, t+\delta]}\frac{\lambda_1(x)}{\lambda_0(x)}\mid \mathcal{F}_t\right]  e^{-\int_t^{t+\delta}\lambda_1(s) - \lambda_0(s)ds} \frac{\mathcal{L}_{t}(\lambda)}{\mathcal{L}_{t}(\lambda_0)}\\
&=e^{\int_t^{t+\delta}\lambda_1(s) - \lambda_0(s)ds}  e^{-\int_t^{t+\delta}\lambda_1(s) - \lambda_0(s)ds}  \frac{\mathcal{L}_{t}(\lambda)}{\mathcal{L}_{t}(\lambda_0)}\\
&=\frac{\mathcal{L}_{t}(\lambda)}{\mathcal{L}_{t}(\lambda_0)}.\\
    \end{split}
\end{equation*}
\end{proof}

\subsubsection{Proof of Theorem \ref{thm:compositemartingale}}
\begin{proof}
    \begin{equation*}
\begin{split}
              M^{\Lambda_0}(t) = \mathbb{E}^\Pi\left[\frac{\mathcal{L}_{t}(\theta)}{\mathcal{L}_{t}(\theta_0)}\right] =\int \frac{\mathcal{L}_t(\theta)}{\mathcal{L}_t(\theta_0)}d\Pi(\theta) =& \frac{\phi^{\phi}}{\Gamma(\phi)} e^{\Lambda_0(t)}\int x^{N(t)}e^{-x\Lambda_0(t)}x^{\phi - 1}e^{-x\phi}dx \\
              =&  \frac{\phi^{\phi}}{(\phi + \Lambda_0(t))^{\phi + N(t)}} \frac{\Gamma(\phi + N(t))}{\Gamma(\phi)} e^{\Lambda_0(t)}
\end{split}
\end{equation*}

Moreover, by Fubini's theorem
\begin{equation*}
    \begin{split}
     \mathbb{E}[M^{\Lambda_0}(t+\delta) | \mathcal{F}_t] =&\mathbb{E}\left[\int \frac{\mathcal{L}_{t+\delta}(\theta)}{\mathcal{L}_{t+\delta}(\theta_0)}d\Pi(\theta) | \mathcal{F}_t\right]\\
     =&\int \mathbb{E}\left[\frac{\mathcal{L}_{t+\delta}(\theta)}{\mathcal{L}_{t+\delta}(\theta_0)} | \mathcal{F}_t \right] d\Pi(\theta)\\
          =&\int \frac{\mathcal{L}_{t}(\theta)}{\mathcal{L}_{t}(\theta_0)}  d\Pi(\theta)\\
     =& M^{\Lambda_0}(t)
\end{split}
\end{equation*}
\end{proof}

\subsubsection{Proof of Theorem \ref{thm:single_hypothesis_test}}
\begin{proof}
    The assumption $\Lambda(t) / t \rightarrow \bar{\lambda}$ implies $N(t) / t \rightarrow \bar{\lambda}$ almost surely by lemma \ref{lemma:stronglaw}. We therefore write $N(t) = \bar{\lambda}t + \epsilon(t)$ where $\epsilon(t)=o_{a.s.}(t)$. Similarly, the assumption $\Lambda_0(t) / t \rightarrow \bar{\lambda}_0$ allows us to write $\Lambda_0(t) = \bar{\lambda}_0 t + \epsilon_0(t)$ where $\epsilon_0(t) = o(t)$. From equation \eqref{eq:composite_martingale}
\begin{equation*}
    \begin{split}
        \frac{\log M^{\Lambda_0}(t)}{t} =& -\frac{(N(t) + \phi)}{t}\log(\Lambda_0(t) + \phi) + \frac{\Lambda_0(t)}{t} + \frac{1}{t}\log\Gamma(N(t)+\phi)+o(1)\\
        =& -(\bar{\lambda} + \frac{\epsilon(t) + \phi}{t})\log( \bar{\lambda}_0 t + \epsilon_0(t) + \phi) + \bar{\lambda}_0  + \frac{1}{t}\log\Gamma(\bar{\lambda}t + \epsilon(t)+\phi)+o_{a.s.}(1).
    \end{split}
\end{equation*}
From lemma \ref{lem:stirling}, Stirling's approximation gives
\begin{equation*}
    \begin{split}
        \frac{\log M^{\Lambda_0}(t)}{t}
        =& -\left(\bar{\lambda} + \frac{\epsilon(t) + \phi}{t}\right)\log(\bar{\lambda}_0 t + \epsilon_0(t) + \phi) + \bar{\lambda}_0  \\
        &+
\left(\bar{\lambda} + \frac{\epsilon(t)+\phi - \frac{1}{2}}{t}\right)\log(\bar{\lambda}t + \epsilon(t)+\phi)
        -\bar{\lambda} + o_{a.s.}(1).
    \end{split}
\end{equation*}
Note for $c>0$, $\log (ct+o_{a.s.}(t)) = \log (ct) + \log(1 + o_{a.s.}(t) / ct) = \log(ct) + o_{a.s.}(1)$, and so we can write
\begin{equation*}
    \begin{split}
        \frac{\log M^{\Lambda_0}(t)}{t}
        =& -\left(\bar{\lambda} + \frac{\epsilon(t) + \phi}{t}\right)\log(\bar{\lambda}_0 t) + \bar{\lambda}_0 +
\left(\bar{\lambda} + \frac{\epsilon(t)+\phi - \frac{1}{2}}{t}\right)\log(\bar{\lambda}t)
        -\bar{\lambda} + o_{a.s.}(1)\\
        =& -\bar{\lambda}\log(\bar{\lambda}_0 t) +
\bar{\lambda} \log(\bar{\lambda}t)
        -(\bar{\lambda}-\bar{\lambda}_0) +\frac{\epsilon(t)+\phi}{t}\log\frac{\bar{\lambda}}{\bar{\lambda}_0} - \frac{1}{2t}\log(\bar{\lambda}t) + o_{a.s.}(1)\\
               =& 
\bar{\lambda} \log\left(\frac{\bar{\lambda}}{\bar{\lambda}_0}\right)
        -(\bar{\lambda}-\bar{\lambda}_0) + o_{a.s.}(1)\\
    \end{split}
\end{equation*}
\end{proof}

\subsection{Proofs for section \ref{sec:testing}}

\subsubsection{Proof of Theorem \ref{thm:e_process}}
Recall $C^\alpha(t) = \{ (L^A, L^B) \in \Rp^2 :M(N^A(t),L^A;\phi)M(N^B(t),L^B;\phi) \leq \alpha^{-1}\}$.
Let $\Theta_0(t) = \{(x,x): x\in\mathbb{R}_{\geq 0}\}$. Claim:
\begin{equation*}
\begin{split}
        p(t)=&\inf \{\alpha : C^\alpha(t)\cap \Theta_0(t)\} = \emptyset\\
        =& 1/E(t)\\
       \text{where}\hspace{1cm} E(t)=& \prod_{i \in \{A, B\}} \frac{\phi^\phi}{(\phi + \hat{\Lambda}(t))^{\phi+N^i(t)}}\frac{\Gamma(\phi+N^i(t))}{\Gamma(\phi)}e^{\hat{\Lambda}(t)},\\
\end{split}
\end{equation*}
and $\hat{\Lambda}(t) = \frac{1}{2}(N^A(t) + N^B(t))$. 
\begin{proof}
Observe that 
\begin{equation*}
    E(t) = \underset{(L^A, L^B) \in \Theta_0}{\inf}M(N^A(t),L^A;\phi)M(N^B(t),L^B;\phi)=M(N^A(t),\hat{\Lambda}(t);\phi)M(N^B(t),\hat{\Lambda}(t);\phi)
\end{equation*}
as $\hat{\Lambda}(t) = \argmin_{L\in\mathbb{R}_{\geq 0}} M(N^A(t),L;\phi)M(N^B(t),L;\phi)$. This is easily verified by setting the first derivative to zero and solving for $x$. $(\hat{\Lambda}(t),\hat{\Lambda}(t))$ is therefore the ``last'' element in $\mathbb{R}_{\geq 0}^2$ to remain in $\Theta_0$ as the upper bound of $\alpha^{-1}$ is decreased. Formally:

We first show $p(t) \leq 1/E(t)$:

    $p(t) = \inf \{\alpha: C^\alpha(t) \cap \Theta_0 = \emptyset \}$ $\Rightarrow$ $\forall\alpha < p(t)$, $C^\alpha(t) \cap \Theta_0(t) \neq \emptyset$  $\Rightarrow$  $\forall\alpha < p(t)$, $\exists (L^A,L^B) \in \Theta_0$, such that $M(N^A(t),L^A;\phi)M(N^B(t),L^B;\phi) \leq \alpha^{-1}$  $\Rightarrow$ $\forall\alpha < p(t)$, $\underset{(L^A, L^B) \in \Theta_0}{\inf}M(N^A(t),L^A;\phi)M(N^B(t),L^B;\phi) \leq \alpha^{-1}$  $\Rightarrow$ $\forall\alpha < p(t)$, $\alpha \leq 1/E(t)$  $\Rightarrow$ $p(t) \leq 1/E(t)$.

Last, we show $p(t) \geq 1/E(t)$:

Conversely $p(t) = \inf \{\alpha: C^\alpha(t) \cap \Theta_0(t) = \emptyset \}$ $\Rightarrow$ $\forall\alpha > p(t)$, $C^\alpha(t) \cap \Theta_0(t) = \emptyset$ 
 $\Rightarrow$ $\forall\alpha > p(t)$, $\forall (L^A, L^B) \in \Theta_0$, $M(N^A(t),L^A;\phi)M(N^B(t),L^B;\phi) > \alpha^{-1}$ $\Rightarrow$ $\forall\alpha > p(t)$, $\underset{(L^A, L^B) \in \Theta_0}{\inf}M(N^A(t),L^A;\phi)M(N^B(t),L^B;\phi)  > \alpha^{-1}$, $\Rightarrow$ $\forall\alpha > p(t)$, $\alpha > 1/E(t) $, $\Rightarrow$ $p(t) \geq 1/E(t)$.
    \end{proof}

\subsubsection{Proof of Theorem \ref{thm:asymp_growth}}
\begin{proof}
Up to terms constant in $t$, $\log E(t)$ can be written as
\begin{equation*}
\begin{split}
     \log E(t) =& -(2\phi + N^A(t)+N^B(t))\log\left(\phi + \frac{1}{2}N^A(t) + \frac{1}{2}N^B(t)\right)\\
     &+\log\Gamma(N^A(t)+\phi)\\
     &+\log\Gamma(N^B(t)+\phi)\\
     &+N^A(t)+N^B(t)\\
     &+const.
\end{split}
\end{equation*}
Assumption $\Lambda^B(t)/t \rightarrow \bar{\lambda}^B$ implies $N^B(t)/t \rightarrow \bar{\lambda}^B$ almost surely by lemma \ref{lemma:stronglaw}, which we write as $N^B(t) = \bar{\lambda}^B t + \epsilon^B(t)$, where $\epsilon^B(t) = o_{a.s.}(t)$. Similarly for $N^A(t)= \bar{\lambda}^A t + \epsilon^A(t)$. Substituting these expressions for $N^A(t)$ and $N^B(t)$ and applying Stirling's approximation in lemma \ref{lem:stirling} yields
\begin{equation*}
\begin{split}
     \log E(t) =& -(\bar{\lambda}^A t+\bar{\lambda}^Bt+\epsilon^A(t)+\epsilon^B(t) + 2\phi)\log\left(\frac{1}{2}\bar{\lambda}^At + \frac{1}{2}\bar{\lambda}^Bt + \frac{1}{2}\epsilon^A(t)+\frac{1}{2}\epsilon^B(t)+\phi\right)\\
     &+(\bar{\lambda}^A t+\epsilon^A(t)+\phi-\frac{1}{2})\log(\bar{\lambda}^A t+\epsilon^A(t)+\phi)- (\bar{\lambda}^A t+\epsilon^A(t)+\phi)\\
     &+(\bar{\lambda}^B t+\epsilon^B(t)+\phi-\frac{1}{2})\log(\bar{\lambda}^B t+\epsilon^B(t)+\phi) - (\bar{\lambda}^B t+\epsilon^B(t)+\phi)\\
     &+\bar{\lambda}^A t+\epsilon^A(t)+\bar{\lambda}^B t+\epsilon^B(t)\\
     &+const.
\end{split}
\end{equation*}
Dividing by $t$ and cancelling terms yields
\begin{equation*}
\begin{split}
     \frac{\log E(t)}{t} =& -\left(\bar{\lambda}^A +\bar{\lambda}^B+\frac{\epsilon^A(t)+\epsilon^B(t) + 2\phi}{t}\right)\log\left(\frac{1}{2}\bar{\lambda}^At + \frac{1}{2}\bar{\lambda}^Bt + \frac{1}{2}\epsilon^A(t)+\frac{1}{2}\epsilon^B(t)+\phi\right)\\
     &+\left(\bar{\lambda}^A +\frac{\epsilon^A(t)+\phi-\frac{1}{2}}{t}\right)\log(\bar{\lambda}^A t+\epsilon^A(t)+\phi)\\
     &+\left(\bar{\lambda}^B +\frac{\epsilon^B(t)+\phi-\frac{1}{2}}{t}\right)\log(\bar{\lambda}^B t+\epsilon^B(t)+\phi)\\
     &+o_{a.s.}(1)
\end{split}
\end{equation*}
Replacing expressions like $\log(\bar{\lambda}^A t + \epsilon^A(t) + \phi)$ as $\log(\bar{\lambda}^A t) + o_{a.s.}(1)$ yields
\begin{equation*}
\begin{split}
     \frac{\log E(t)}{t} =& -\left(\bar{\lambda}^A +\bar{\lambda}^B+\frac{\epsilon^A(t)+\epsilon^B(t) + 2\phi}{t}\right)\log\left(\frac{1}{2}\bar{\lambda}^At + \frac{1}{2}\bar{\lambda}^Bt\right)\\
     &+\left(\bar{\lambda}^A +\frac{\epsilon^A(t)+\phi-\frac{1}{2}}{t}\right)\log(\bar{\lambda}^A t)\\
     &+\left(\bar{\lambda}^B +\frac{\epsilon^B(t)+\phi-\frac{1}{2}}{t}\right)\log(\bar{\lambda}^B t)\\
     &+o_{a.s.}(1)
\end{split}
\end{equation*}
The $\frac{\epsilon^A(t)}{t}\log t$ and $\frac{\epsilon^B(t)}{t}\log t$ terms cancel, with other terms being absorbed into the $o_{a.s.}(1)$ term, yielding
\begin{equation*}
\begin{split}
     \frac{\log E(t)}{t} =& -\left(\bar{\lambda}^A +\bar{\lambda}^B\right)\log\left(\frac{1}{2}\bar{\lambda}^A + \frac{1}{2}\bar{\lambda}^B\right)+\bar{\lambda}^A\log(\bar{\lambda}^A )+\bar{\lambda}^B\log(\bar{\lambda}^B )+o_{a.s.}(1)\\
=&\bar{\lambda}^A\log \frac{2\bar{\lambda}^A}{\bar{\lambda}^A+\bar{\lambda}^B} + \bar{\lambda}^B\log \frac{2\bar{\lambda}^B}{\bar{\lambda}^A+\bar{\lambda}^B}\hspace{1cm} +o_{a.s.}(1)
\end{split}
\end{equation*}

\end{proof}

\subsubsection{Proof of Theorem \ref{thm:asymp_growth2}}
\begin{proof}
Assumption $\Lambda^B(t)/t \rightarrow \bar{\lambda}^B$ implies $N^B(t)/t \rightarrow \bar{\lambda}^B$ almost surely by lemma \ref{lemma:stronglaw}, which we write as $N^B(t) = \bar{\lambda}^B t + \epsilon^B(t)$, where $\epsilon^B(t) = o_{a.s.}(t)$. Similarly for $N^A(t)= \bar{\lambda}^A t + \epsilon^A(t)$. Substituting these expressions for $N^A(t)$ and $N^B(t)$ and applying Stirling's approximation in lemma \ref{lem:stirling} yields

\begin{equation}
\begin{split}
        \frac{\log\tilde{E}(t)}{t} =& \left(\bar{\lambda}^A + \frac{\epsilon^A(t) + \beta}{t}\right)\log(\bar{\lambda}^A t + \epsilon^A(t) + \beta) \\
        &+ \left(\bar{\lambda}^B + \frac{\epsilon^B(t) + \alpha}{t}\right)\log(\bar{\lambda}^B t + \epsilon^B(t) + \alpha)\\
        &- \left(\bar{\lambda}^A+ \bar{\lambda}^B + \frac{\alpha+\beta+\epsilon^A(t)   + \epsilon^B(t)}{t}\right)\log(\alpha+\beta+\bar{\lambda}^A t + \epsilon^A(t) + \bar{\lambda}^B t + \epsilon^B(t))\\
        &+ (\bar{\lambda}^A + \bar{\lambda}^B)\log 2\\
        &+ o_{a.s.}(1).
\end{split}
\end{equation}
The first logarithmic expression can be written
\begin{equation*}
    \begin{split}
        \log(\bar{\lambda}^A t + \epsilon^A(t) + \beta) =& \log(\bar{\lambda}^A t) + \log\left(1 + \frac{\epsilon^A(t) + \beta}{\bar{\lambda}^A t}\right)\\
        =&\log(\bar{\lambda}^A t) + o_{a.s.}(1),
    \end{split}
\end{equation*}
and similarly for the second and third, yielding
\begin{equation}
\begin{split}
        \frac{\log\tilde{E}(t)}{t} =& \left(\bar{\lambda}^A + \frac{\epsilon^A(t) + \beta}{t}\right)\log(\bar{\lambda}^At) \\
        &+ \left(\bar{\lambda}^B + \frac{\epsilon^B(t) + \alpha}{t}\right)\log(\bar{\lambda}^B t)\\
        &- \left(\bar{\lambda}^A+ \bar{\lambda}^B + \frac{\alpha+\beta+\epsilon^A(t)   + \epsilon^B(t)}{t}\right)\log(\bar{\lambda}^A t+\bar{\lambda}^B t)\\
        &+ (\bar{\lambda}^A + \bar{\lambda}^B)\log 2\\
        &+ o_{a.s.}(1).
\end{split}
\end{equation}
The $\frac{\epsilon^A(t)}{t}\log t$ and $\frac{\epsilon^B(t)}{t}\log t$ terms cancel
\begin{equation*}
\begin{split}
        \frac{\log\tilde{E}(t)}{t} =& \left(\bar{\lambda}^A + \frac{\epsilon^A(t) + \beta}{t}\right)\log(\bar{\lambda}^A) \\
        &+ \left(\bar{\lambda}^B + \frac{\epsilon^B(t) + \alpha}{t}\right)\log(\bar{\lambda}^B)\\
        &- \left(\bar{\lambda}^A+ \bar{\lambda}^B + \frac{\alpha+\beta+\epsilon^A(t)   + \epsilon^B(t)}{t}\right)\log(\bar{\lambda}^A +\bar{\lambda}^B )\\
        &+ (\bar{\lambda}^A + \bar{\lambda}^B)\log 2\\
        &+ o_{a.s.}(1)\\
        =& \bar{\lambda}^A \log(\bar{\lambda}^A)+\bar{\lambda}^B\log(\bar{\lambda}^B)
        -\left(\bar{\lambda}^A+ \bar{\lambda}^B\right)\log(\bar{\lambda}^A+\bar{\lambda}^B)+ (\bar{\lambda}^A + \bar{\lambda}^B)\log 2 + o_{a.s.}(1)\\
         =&\bar{\lambda}^A\log \frac{2\bar{\lambda}^A}{\bar{\lambda}^A+\bar{\lambda}^B} + \bar{\lambda}^B\log \frac{2\bar{\lambda}^B}{\bar{\lambda}^A+\bar{\lambda}^B}+o_{a.s.}(1)
\end{split}
\end{equation*}

\end{proof}

\subsubsection{Proof of Theorem \ref{thm:asymp_growth3}}
\begin{proof}
    Assumption $\Lambda^B(t)/t \rightarrow \bar{\lambda}^B$ implies $N^B(t)/t \rightarrow \bar{\lambda}^B$ almost surely by lemma \ref{lemma:stronglaw}, which we write as $N^B(t) = \bar{\lambda}^B t + \epsilon^B(t)$, where $\epsilon^B(t) = o_{a.s.}(t)$. Similarly for $N^A(t)= \bar{\lambda}^A t + \epsilon^A(t)$, where $\epsilon^A(t) = o_{a.s.}(t)$.

\begin{equation*}
    \begin{split}
        \frac{\log E^A(t)}{t} &= -\frac{1}{2t}\log(\phi + t) + \frac{1}{2}\frac{1}{\phi + t}\frac{(\bar{\lambda}^B t - \bar{\lambda}^A t +\epsilon^B(t) - \epsilon^A(t))^2}{\bar{\lambda}^B t + \bar{\lambda}^A t +\epsilon^B(t) + \epsilon^A(t)} + const.\\
        &=o_{a.s.}(1) + \frac{1}{2}\frac{1}{1 + \phi / t}\frac{(\bar{\lambda}^B  - \bar{\lambda}^A  + o_{a.s.}(1))^2}{\bar{\lambda}^B + \bar{\lambda}^A + o_{a.s.}(1)}\\
        &=\frac{1}{2}\frac{(\bar{\lambda}^B  - \bar{\lambda}^A)^2}{\bar{\lambda}^B + \bar{\lambda}^A} + o_{a.s.}(1)\\
    \end{split}
\end{equation*}
\end{proof}

\section{Numerical Computation of Confidence Processes}
\subsection{Numerical Computation of Confidence Processes for $\lambda^A(t)$ and $\lambda^B(t)$}
We present the numerical method for computing the confidence process of $\Lambda^B(t)$ only. Recall from equation \eqref{eq:joint_confidence} that the set $C^\alpha(t)$ comprises the set of elements $(x_1, x_2) \in \mathbb{R}^2_{\geq 0}$ satisfying
\begin{equation*}
    \frac{\phi^\phi}{(\phi + x_2)^{\phi+N^B(t)}}\frac{\Gamma(\phi+N^B(t))}{\Gamma(\phi)}e^{x_2} \leq \alpha^{-1} \left(\frac{\phi^\phi}{(\phi + x_1)^{\phi+N^A(t)}}\frac{\Gamma(\phi+N^A(t))}{\Gamma(\phi)}e^{x_1}\right)^{-1}.
\end{equation*}
The term on the right-hand side of the inequality obtains a maximum when $x_1 = N^A(t)$.
Define a constant
    \begin{equation*}
        k^A(t) = \frac{1}{\alpha\phi^{2\phi}}\frac{(\phi + N^A(t))^{\phi+N^A(t)}}{e^{N^A(t)} }\frac{\Gamma(\phi)}{\Gamma(\phi+N^A(t))}\frac{\Gamma(\phi)}{\Gamma(\phi+N^B(t))},
    \end{equation*}
and a function
\begin{equation*}
                    f^B(x) =\frac{e^x}{(\phi + x)^{\phi+N^B(t)}}- K^A(t).\\
\end{equation*}
Then $\{y > 0 : (x,y) \in C^\alpha(t) \text{ for some } x\} = \{y : f^B(y) \leq 0 \}$. The extrema of the interval defining the confidence process for $\Lambda^B(t)$ are simply the roots of $f^B(x)$. The result for $\Lambda^A(t)$ is similar.

\subsection{Numerical Computation of Confidence Processes for $\Lambda^B(t) - \Lambda^A(t)$}
The transformed confidence set is given by
\begin{equation*}
    T(C^\alpha(t)) = \left\{(w, v) \in \mathbb{R}\times\Rp : g(w,v) \leq -\log \alpha\right\},
\end{equation*}
where
\begin{equation*}
\begin{split}
    g(w,v) =& v - (\phi+N^B(t))\log\left(\phi + \frac{1}{2}\left(v+w\right)\right)- (\phi+N^A(t))\log\left(\phi + \frac{1}{2}\left(v-w\right)\right)\\
     &+ 2\phi \log \phi - 2 \log \Gamma(\phi) + \log \Gamma(\phi + N^A(t)) + \log \Gamma(\phi + N^B(t) ).\\
\end{split}
\end{equation*}
The interval for $\Lambda^B(t) - \Lambda^A(t)$ is the projection $\{w : (w,v) \in T(C^\alpha(t)) \text{ for some } v\}$, and we seek the lower and upper bounds that define this interval. These can be obtained by maximizing and minimizing the function $f(w,v) = w$ over $T(C^\alpha(t))$. Let $(w_u, v_u)$ and $(w_l, v_l)$ denote the solutions to the maximization and minimization problems respectively. Now consider the two equations
\begin{equation}
\label{eq:optimization_equations}
\begin{split}
        g(w, v) &= \alpha^{-1},\\
        \frac{\partial g(w,v)}{\partial v} &= 0
\end{split}
\end{equation}
The second equation allows $v$ to be expressed in terms of $w$ as
\begin{equation*}
        v = h(w) = \frac{1}{2}(N^A(t) + N^B(t)) - \phi + \sqrt{a + b},
\end{equation*}
where
\begin{equation*}
\begin{split}
    a &= \frac{1}{4}N^A(t)^2 + \frac{1}{2}N^A(t)N^B(t) + N^A(t)\phi + N^A(t)w, \\
    b &= \frac{1}{4}N^B(t)^2 + N^B(t)\phi - N^B(t)w + \phi^2 + w^2. \\
\end{split}
\end{equation*}
The solution for $w$ can then be found by solving $g(w, h(w)) = - \log(\alpha)$.

The solutions to the optimization problem $(w_l, v_l)$ and $(w_u, v_u)$ must satisfy equations \eqref{eq:optimization_equations} as the solutions live on the boundary of $T(C^\alpha(t))$, satisfying the first equation, and $\nabla g$ has zero $v$ component, satisfying the second equation. The values $w_l$ and $w_u$ are then the two roots of the function $g(w, h(w)) + \log \alpha $.

\section{Supplemental Figures}
Figure \ref{fig:intensity} shows the intensity function and the realized point process from example \ref{ex:single_process}. Figure \ref{fig:intensities} shows the two intensity functions and realized point processes from example \ref{ex:two_processes}.
\begin{figure}[h]
    \centering
    \includegraphics[width=\linewidth]{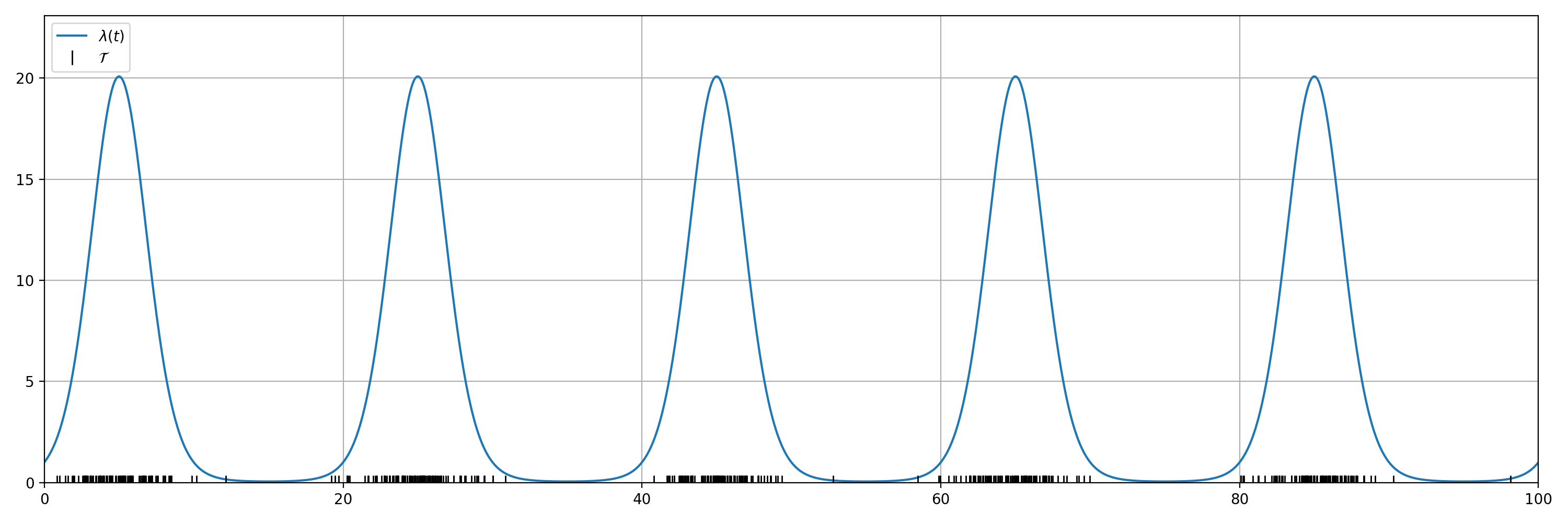}
    \caption{Intensity function $\lambda(t)$ (blue) with realized point process $\mathcal{T}$ (black ticks)}
    \label{fig:intensity}
\end{figure}

\begin{figure}[h]
    \centering
    \includegraphics[width=\linewidth]{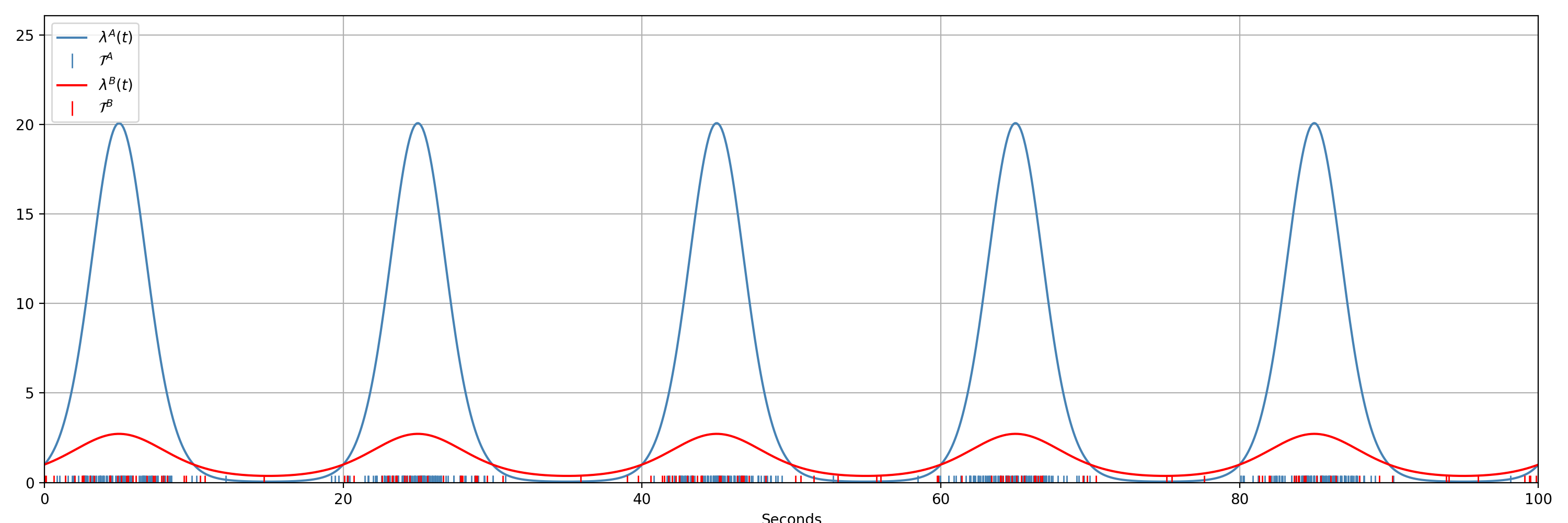}
    \caption{$\mathcal{T}^A \sim \mathcal{P}(\lambda^A)$ and $\mathcal{T}^B \sim \mathcal{P}(\lambda^B)$ with $\lambda^A(t) =e^{3\sin(2\pi t/20)}$ and $\lambda^B(t) =e^{2\sin(2\pi t/20)}$}
    \label{fig:intensities}
\end{figure}

    \end{document}